\documentclass[11pt]{amsart}
\usepackage{latexsym,amssymb,amsmath,amscd}
\input epsf
\def\boxit#1{\vbox{\hrule height1pt\hbox{\vrule width1pt\kern3pt
  \vbox{\kern3pt#1\kern3pt}\kern3pt\vrule width1pt}\hrule height1pt}}

\def\BC{\mathbb C}
\def\BQ{\mathbb Q}
\def\BP{\mathbb P}
\def\BN{\mathbb N}

\def\fgl{\mathfrak g\mathfrak l}

\def\la{\lambda}

\def\tdim{\rm dim}
\def\hd{,...,}

\def\be{\begin{equation}}
\def\ene{\end{equation}}
\def\aaa{{\bold a}}\def\bbb{{\bold b}}

\def\inv{{}^{-1}}

\def\cB{{\mathcal B}}

\def\cS{{\mathcal S}}

\def\cW{{\mathcal W}}
\def\cO{{\mathcal O}}
\def\CC{\mathbb C}

\def\11{\mathbf 1}
\def\PP{\mathbb P}
\def\BZ{\mathbb Z}

\def\FS{{\mathfrak S}}

\def\l{\lambda}
\def\a{\alpha}

\def\o{\omega}

\def\b{\beta}
\def\g{\gamma}
\def\s{\sigma}

\def\d{\delta}

\def\e{\varepsilon}
\def\ot{{\mathord{\,\otimes }\,}}
\def\op{{\mathord{\,\oplus }\,}}
\def\otc{{\mathord{\otimes\cdots\otimes}\;}}

\def\ra{{\mathord{\;\rightarrow\;}}}

\def\tdim{{\rm dim}\;}\def\tcodim{{\rm codim}\;}

\def\tdet{{\rm det}}
\def\tperm{{\rm perm}}
\def\La#1{\Lambda^{#1}}
\def\tmult{{\rm mult}}
\def\vp{\bold {VP}}\def\vnp{\bold {VNP}}

\def\bl{\hbox{\boldmath$\lambda$\unboldmath}}
 \def\bpi{\hbox{\boldmath$\pi$\unboldmath}}

\newtheorem{theo}{Theorem}[section]

\newtheorem{prop}[theo]{Proposition}
\newtheorem{conjecture}[theo]{Conjecture}

\newtheorem{theorem}{Theorem}[section]
\newtheorem{proposition}[theorem]{Proposition}
\newtheorem{lemma}[theorem]{Lemma}
\newtheorem{corollary}[theorem]{Corollary}

\theoremstyle{definition}
\newtheorem{definition}[theorem]{Definition}

\newtheorem{question}[theorem]{Question}
\newtheorem{example}[theorem]{Example}

\newtheorem{problem}[theorem]{Problem}

\theoremstyle{remark}
\newtheorem{remark}[theorem]{Remark}

\numberwithin{equation}{subsection}
 \numberwithin{theorem}{subsection}

\def\ol{\overline}

\def\wsL{L_{\mathrm{ws}}}
\def\Oh{\mathcal{O}}
\def\VPe{\vp_e}
\def\VPws{\vp_{\mathrm{ws}}}
\def\VP{\vp}
\def\VNP{\vnp}
\def\N{\mathbb{N}}
\def\det{\mathrm{det}}
\def\per{\mathrm{per}}
\def\bL{\underline{L}}
\def\bwsL{\underline{L}_{\mathrm{\,ws}}}
\def\bVPws{\overline{\vp_{\,\mathrm{ws}}}}
\def\bVP{\overline{\vp}}

\def\GL{GL}
\def\Mat{\mathrm{Mat}}

\def\Hom{\mathrm{Hom}}
\def\peterchange{}

\begin{document}
\title[GCT approach to $\vp\neq\vnp$]{An overview of mathematical issues arising in  the Geometric complexity theory approach to $\vp\neq \vnp$}
\author{Peter B\"urgisser,  J.M.~Landsberg, Laurent Manivel and Jerzy Weyman}
\begin{abstract}We discuss the geometry of orbit closures and  the asymptotic behavior of Kronecker coefficients
in the context of the Geometric Complexity Theory program to prove
a variant of Valiant's  algebraic analog of  the $P\neq NP$ conjecture.
  We also describe  the precise separation of   complexity classes   that their   program proposes to demonstrate.
\end{abstract}
 \thanks{B\"urgisser supported by DFG-grants BU 1371/2-1  and BU 1371/3-1. Landsberg, Weyman respectively supported by NSF grants DMS-0805782 and DMS-0901185}
\email{pbuerg@upb.de,jml@math.tamu.edu,Laurent.Manivel@ujf-grenoble.fr,j.weyman@neu.edu}
\maketitle

\section{Introduction}
In a series of   papers \cite{MS1,MS2,MS3,MS4,MS5,MS6,MS7,MS8},
K.~Mulmuley and M.~Sohoni outline an approach to the $P$ v.s.\ $NP$ problem, that they call the
{\it Geometric Complexity Theory (GCT)} program. The starting point is
Valiant's conjecture 
\cite{vali:79-3} (see also~\cite{MR922386,buer:00-3})
that the permanent hypersurface in $m^2$ variables
(i.e., the set of $m\times m$ matrices $X$ with $\tperm_m(X)=0$) cannot be realized as
an affine linear section
of the determinant hypersurface in $n(m)^2$ variables with $n(m)$ a polynomial function of $m$.
Their program (at least up to \cite{MS2}) translates the problem
of proving Valiant's conjecture to proving a conjecture in
representation theory. In this paper we give an exposition of the program outlined in \cite{MS1,MS2},
present the representation-theoretic conjecture in
detail,  and present a framework for reducing their
representation theory  questions
to easier questions   by taking more geometric information into account. We also
precisely identify the  complexity problem the GCT approach proposes to solve and how it compares to Valiant's
original conjecture, and discuss related issues
in geometry that arise from their program.
The goal of this paper  is to clarify the state of the art, and
identify steps that would further advance the program using recent
advances in geometry and representation theory.

The GCT program translates the study of the hypersurfaces
$$\{\tperm_m=0\}\subset \BC^{m^2}\ \ {\rm and}\ \ \{\tdet_n=0\}\subset\BC^{n^2},
$$
to a study of the orbit closures
$$
\ol{GL_{n^2} \cdot [\ell^{n-m}\tperm_m]}\subset \BP (S^n\BC^{n^2}) \ \ {\rm  and}\ \
\ol{GL_{n^2} \cdot [\tdet_n}]\subset \BP (S^n\BC^{n^2}),
$$
where $S^n\BC^{n^2}$ denotes the space
of homogeneous polynomials of degree $n$ in $n^2$ variables. Here $\ell$ is a linear coordinate
on $\BC$, and one takes any linear inclusion $\BC\op \BC^{m^2}\subset \BC^{n^2}$ to have
$\ell^{n-m}\tperm_m$ be
a homogeneous degree $n$ polynomial  on  $\BC^{n^2}$. Mulmuley and Sohoni
observe that a variant of Valiant's hypothesis would be proved if one could show:

\begin{conjecture}\cite{MS1}\label{msmainconj}
There does not exist a constant $c\ge 1$
such that for sufficiently large $m$,
$$\ol{GL_{m^{2c}} \cdot  [\ell^{m^c-m}\tperm_m]}\subset \ol{GL_{m^{2c}} \cdot [\tdet_{m^c}]}.$$
\end{conjecture}

It is known that
$\ol{GL_{n^2} \cdot [\ell^{n-m}\tperm_m]} \subset \ol{GL_{n^2} \cdot [\tdet_n]}$
for $n =\cO(m^2 2^m)$,
see Remark~\ref{re:ryser}.

For a closed subvariety $X$ of $\BP  V$, let $\hat X\subset V$ denote the cone over $X$.
Let $I(\hat X)\subset Sym(V^*)$ be the ideal of polynomials vanishing on $\hat X$,
and let $\BC[X]=Sym(V^*)/I(\hat X)$ denote the homogeneous coordinate ring.
For two closed subvarieties $X,Y$ of $\BP V$, one has $X\subset Y$ iff $\BC[Y]$ surjects onto $\BC[X]$
by restriction of polynomial functions.

The GCT program sets out to prove:

\begin{conjecture}\cite{MS1}\label{msmainconjb}
For all $c \ge 1$ and for infinitely many $m$ there exists
an irreducible $GL_{m^{2c}}$-module appearing in
$\BC[\ol{GL_{m^{2c}} \cdot   [\ell^{m^c-m}\tperm_m}]]$,
but not appearing in $\BC[\ol{GL_{m^{2c}} \cdot  [\tdet_{m^c}}]]$.
\end{conjecture}

Both varieties occuring in Conjecture~\ref{msmainconjb}
are invariant under $GL_{m^{2c}}$, so their coordinate
rings are $GL_{m^{2c}}$-modules.
Conjecture~\ref{msmainconj} is a straightforward
consequence of Conjecture \ref{msmainconjb} by Schur's lemma.

A program to prove Conjecture \ref{msmainconjb} is outlined in~\cite{MS2},
which also contains a discussion why the desired
irreducible modules (called {\em representation theoretic obstructions})
should exist.
This is closely related to a separability question \cite[Conjecture 12.4]{MS2}
that we will not address in this paper.

There are several paths one could take to try to find such a sequence of modules. The path chosen
in \cite{MS2} is to consider $SL_{n^2}\cdot \tdet_n$ and $SL_{m^2}\cdot \tperm_m$ because on one hand,
their coordinate rings can be determined in principle using representation theory, and on the other hand,
they are closed affine varieties. Mulmuley and Sohoni
observe that any irreducible $SL_{n^2}$-module appearing in $\BC[SL_{n^2}\cdot \tdet_n]$ must
also appear in the degree~$\d$ part of the
graded $SL_{n^2}$-module  $\BC[\ol{GL_{n^2}\cdot [\tdet_n}]]_\delta$ for some~$\delta$.
Regarding the permanent, for $n>m$,  $SL_{n^2}\cdot\ell^{n-m}\tperm_m$ is not closed, so they develop
machinery to transport information about $\BC[SL_{m^2}\cdot \tperm_m]$ to $\BC[\ol{GL_{n^2}\cdot
[\ell^{n-m}\tperm_m}]]$,
in particular they introduce a notion of {\it partial stability}.

We make a close study of how one
might 
exploit partial stability to determine
the $GL_{n^2}$-module decomposition of $\BC[\ol{GL_{n^2}\cdot [\ell^{n-m}\tperm_m}]]$ in \S\ref{examplesect}.
We also discuss a more elementary approach to studying which modules in
$\BC[ GL_{n^2}\cdot [\ell^{n-m}\tperm_m]]$ could appear in
the degree~$\d$ part of
$\BC[\ol{GL_{n^2}\cdot [\ell^{n-m}\tperm_m}]]$. 
One could get more information from the elementary approach if one
could solve the {\it extension problem} of determining which functions on the orbit
${GL_{n^2}\cdot [\ell^{n-m}\tperm_m]}$ extend to the orbit closure $\ol{GL_{n^2}\cdot [\ell^{n-m}\tperm_m]}$.
In general the extension problem is very difficult, we discuss it in \S\ref{orbitclosuresect}.

We express the restrictions on modules appearing in $\BC[\ol{GL_{n^2}\cdot [\ell^{n-m}\tperm_m}]]$
that we do have, as well as our information regarding $\BC[\ol{GL_{n^2}\cdot  [\tdet_n}]]$,
in terms of {\it Kronecker coefficients} \peterchange
and {\it symmetric Kronecker coefficients} that we introduce in \S\ref{se:1stmainex}.
Kronecker coefficients are defined as the multiplicities occurring 
in tensor products of representations of symmetric groups.
We review all relevant information regarding these
coefficients that we are aware of in \S\ref{kronsect}. Unfortunately, from this
information, we are currently unable to see how one could prove Conjecture \ref{msmainconjb} in
the case $c=1$ (which is straight-forward by other means), let alone for all~$c$.
Nevertheless, we have found the GCT program
a beautiful source of inspiration for future work.


This program is beginning to gain the attention of the mathematical community, for
example the recent preprints   \cite{Popovpreprint}, where an algorithm
is given for determining if one orbit is in the closure of another, and
\cite{BORpreprint}, where a conjecture of Mulmuley regarding Kronecker
coefficients is disproven and, in an appendix by Mulmuley, a  modified conjecture is proposed.
Since the original submission of this paper in July 2009, there have
been several developments \cite{LMRdet, buci:09,buci:10,kumar:10,buic:10}
whose relevance we note where appropriate in the body of the paper.

\subsection*{Acknowledgments} It is a pleasure to thank Shrawan Kumar for
very useful discussions. This paper is an outgrowth of the AIM workshop
{\it Geometry and representation theory of tensors for computer science, statistics and other areas}
July 21-25, 2008, and authors gratefully thank AIM and the other participants of the workshop.
We also thank the anonymous referees for their useful suggestions.

\section{Overview}

We begin,  in \S\ref{npsect},  by establishing notation and reviewing basic facts from representation theory
that we use throughout.
In \S\ref{orbitcoordringsect} we discuss coordinate rings of orbits and orbit closures
and in  \S\ref{examplesect} we make a detailed study of the cases at hand.
In \S\ref{pstabvalueb} we state the theorems in \cite{MS2}
and also give an overview of their proofs.
The consequences of partial stability can be viewed
from the perspective of the {\it collapsing method} for computing coordinate rings (and syzygies),
which we discuss in~\S\ref{basicthmsubsect}.

While \cite{MS2} is  primarily concerned with  $SL_{n^2}\cdot\tdet_n$ and a corresponding
closed orbit related to the permanent, we also study the coordinate rings of the orbits of
the general linear group $GL_{n^2}$.
The $GL_{n^2}$-orbits have the disadvantage of not being closed in general,
so one must deal with the {\em extension problem},
which we discuss in \S\ref{orbitclosuresect},
but they have the advantage of having a graded coordinate ring.
%

In the studies of   the coordinate rings of permanent and determinant
{\it Kronecker coefficients} play a central role.
We discuss what is known about the relevant Kronecker coefficients  in \S\ref{kronsect}.
In~\S\ref{complexityclasssect},
we give a brief outline of the relevant algebraic complexity theory involved here.
We explain Valiant's conjecture $\vp\neq \vnp$, how this precisely relates
to the conjecture regarding projecting the determinant to the permanent,
and we formulate Conjecture~\ref{msmainconj}
as the separation of complexity classes $\ol{\VPws}\neq \vnp$.

%

\section{Notation and Preliminaries}\label{npsect}

Throughout we work over the complex numbers $\BC$.
Let $V$ be a complex vector space, let $GL(V)$ denote the general linear group of $V$,
let $v \in V$ and let $G\subseteq GL(V)$ be a subgroup.
We let $G\cdot v\subset V$ denote the orbit of $v$, $\ol{G\cdot v}\subset V$ its
Zariski closure, and $G(v)\subset G$ the stabilizer of~$v$,
so $G\cdot v\simeq G/G(v)$.
Write  $\BC[G\cdot v]$ (respectively
$\BC[\ol{G\cdot v}]$) for the ring  of regular functions on $G\cdot v$ (resp. $\ol{G\cdot v}$).
By restriction, there is  a surjective map $Sym(V^*)\rightarrow \BC[\ol{G\cdot v}]$.

It will be convenient to switch  back and forth between vector spaces and projective spaces.
$\BP V$ denotes the space of lines through the origin in $V$. If $v\in V$ is nonzero,
let $[v]\in \BP V$ denote the corresponding point in projective space, and if
$x\in \BP V$,   let $\hat x\subset V$ denote the corresponding line.
A linear action of $G$ on $V$ induces an action of $G$ on $\BP V$,   let $G([v])$
denote the stabilizer of $[v]\in \BP V$. If $Z\subset \BP V$ is a subset,   let
$\hat Z\subset V$ denote the corresponding cone in $V$.

We will be concerned with the space of homogeneous polynomials
of degree $n$ in $n^2$ variables, $V=S^n(\Mat_{n\times n}^*)=S^nW$.
Here $\Mat_{n\times n}$ denotes the space of $n\times n$-matrices,
$S^nW$ the space of homogeneous polynomials of degree $n$ on $W^*$,
and $G=GL(W)$.
Our main points of interest will be $x=[\tdet_n]$ and $x=[\ell^{n-m}\tperm_m]$,
where $\tdet_n\in S^n(\Mat_{n\times n}^*)$
is the determinant of an $n\times n$ matrix,   $\tperm_m\in S^m(\Mat_{m\times m}^*)$ is
the permanent, we have made a linear inclusion $\Mat_{m\times m}\subset \Mat_{n\times n}$,
and $\ell$ is a linear form on $\Mat_{n\times n}$ annihilating the image of $\Mat_{m\times m}$.


For a reductive group $G$, the set of dominant integral weights $\Lambda^+_G$ indexes the irreducible (finite dimensional) $G$-modules
(see, e.g., \cite{FH,MR1920389}), and for
$\l\in\Lambda_G^+$,    $V_{\l}(G)$ denotes the irreducible $G$-module with
highest weight $\l$, and if $G$ is understood, we just write $V_{\l}$. If $H\subset G$
is a subgroup, and $V$ a $G$-module,
let $V^H:=\{ v\in V\mid \forall h\in H\ h\cdot v=v\}$ denote
the space of $H$-invariant vectors.
For a $G$-module $V$,   let  $\tmult(V_\lambda (G), V)$
denote the multiplicity of the irreducible representation $V_\lambda (G)$ in $V$.

The weight lattice $\Lambda_{GL_M}$ of $GL_M$ is $\BZ^M$ and the dominant
integral weights $\Lambda_{GL_M}^+$  can be identified with
the $M$-tuples $(\pi_1\hd \pi_M)$ with
$\pi_1\geq \pi_2\geq\cdots \geq \pi_M$. For
future reference, we note
\begin{equation}\label{gldual}
V_{(\pi_1\hd \pi_M)}(GL_M)^*=V_{(-\pi_M\hd -\pi_1)}(GL_M).
\end{equation}
The polynomial irreducible representations of $GL_M$ are the Schur modules
$S_{\pi}\BC^M$, indexed
by partitions $\pi=(\pi_1\hd \pi_M)$ with $\pi_1\geq \pi_2\geq\cdots \geq \pi_M\ge 0$.
To get all the rational
irreducible  representations we need to twist by negative powers of the
determinant. This
introduces some redundancies since $S_{\pi}\BC^M\otimes (\det\BC^M)^{\ot
k}=S_{\pi+(k,\ldots ,k)}\BC^M$.
To avoid them, we consider the modules $S_{\pi}\BC^M\otimes (\det\BC^M)^{\ot
k}$ with
$k\in \BZ$ and $\pi=(\pi_1\hd \pi_{M-1},0)$. Moreover we write our partitions as
$\pi=(\pi_1\hd \pi_N)$
with the convention that $\pi_1\geq \cdots \geq \pi_N>0$, and we let
$|\pi|=\pi_1+\cdots + \pi_N$ and $\ell(\pi)=N$.   
We also write $\pi\vdash_m d$ to express that $\pi$ is a partition 
of size $|\pi|=d$ and such that $\ell(\pi) \le m$.\peterchange
The notation
$\pi\mapsto\pi'$ means that
$\pi_1\ge\pi'_1\ge \pi_2\ge\pi'_2\ge\cdots \ge 0.$

The irreducible $SL_M$-modules are obtained by restricting the irreducible
$GL_M$-modules, but
beware that this is insensitive to a twist by the determinant. The weight
lattice of
$\Lambda_{SL_M}$ of $SL_M$ is $\BZ^{M-1}$ and the dominant integral weights
$\Lambda_{SL_M}^+$
are the non-negative linear combinations of the fundamental weights  $\o_1,\ldots,
\o_{M-1}$.
A Schur module $S_{\pi}\BC^M$
considered as an $SL_M$-module has highest weight
$$\l=\bl(\pi) =(\pi_1-\pi_2)\o_1+(\pi_2-\pi_3)\o_2+\cdots + (\pi_{M-1}-\pi_M)\o_{M-1}.$$
We write $S_{\pi}\BC^M=V_{\bl(\pi)}(SL_M)$ or simply $V_{\bl(\pi)}$ if $SL_M$
is clear from the context.

Let $\bpi(\l)$   denote  the smallest partition such that the $GL_M$-module
$S_{\bpi(\l)}\BC^M$,
considered as an $SL_M$-module, is $V_{\l}$.
That is,  $\bpi$ is a  map from $\Lambda^+_{SL_M}$ to $\Lambda^+_{GL_M}$,
mapping $\l=\sum_{j=1}^{M-1}\l_j\o_j$ to
$$\bpi(\l)=(\sum_{j=1}^{M-1}\l_j,
\sum_{j=2}^{M-1}\l_j\hd \l_{M-1}).$$

\section{Stabilizers and coordinate rings of orbits}\label{orbitcoordringsect}

As mentioned in the introduction, \cite{MS2} proposes to study the rings of
regular functions on
$ \ol{GL_{n^2}\cdot\tdet_n} $ and $ \ol{GL_{n^2}\cdot \ell^{n-m}\tperm_m}$
by first studying the regular functions on the
closed orbits $SL_{n^2}\cdot \tdet_n$ and $SL_{m^2}\cdot \ell^{n-m}\tperm_m$.
In this section we review
facts about the coordinate ring of a homogeneous space and stability of orbits, record
observations in \cite{MS2} comparing closed $SL(W)$-orbits and $GL(W)$-orbit closures,
state their definition of partial stability and record Theorem \ref{speccasefirstth}
which illustrates a potential utility of partial stability.

Throughout this section, unless otherwise specified, $G$ will denote
a reductive group and $V$ a $G$-module.

\subsection{Coordinate rings of homogeneous spaces}

The coordinate ring of a reductive group~$G$ has a left-right decomposition, as a $(G-G)$-bimodule,
\be\label{GGmodule}\BC[G]= \bigoplus_{ \l\in \Lambda^+_G  } V_{\l}^*\ot V_{\l},
\ene
\noindent where $V_{\l}$ denotes the irreducible $G$-module of highest weight $\l$.

Let $H\subset G$ be a closed subgroup.
The coordinate ring of the homogeneous space $G/H$ is obtained
by taking (right) $H$-invariants in \eqref{GGmodule} giving rise to  the (left) $G$-module decomposition
\def\tcodim{{\rm codim}\;}
\be\label{LRdecompeqn}\BC[G/H]=\BC[G]^H= \bigoplus_{ \l\in \Lambda^+_G  } V_{\l}^*\ot V_{\l}^H
= \bigoplus_{ \l\in \Lambda^+_G  } (V_{\l}^*)^{\op \tdim  V_{\l}^H}.
\ene
The second equality holds because $V_{\l}^H$ is a trivial (left) $G$-module.
See \cite[Thm. 3, Ch. II, \S 3]{MR768181},  or \cite[\S 7.3]{MR2265844}   for an exposition of these facts.

\subsection{Orbits with reductive stabilizers}\label{staborbitsect}

Let $G$ be a reductive group, let  $V$ be an irreducible $G$-module,
and let $v\in V$ be such that its stabilizer $G(v)$ is reductive.
Then   $G\cdot v= G/G(v)\subset V$ is an affine variety  \cite[Cor. p. 206]{MR0109854}. \peterchange 
The complement of an affine variety in a complete variety is always of pure codimension one 
(see \cite{MR0282977}, chapter 2, Proposition 3.1).
From this it follows  that the boundary of $G\cdot v$ is empty or has pure codimension one
in $\overline{G\cdot v}$. 
Indeed, we can complete~$V$ 
by a hyperplane at infinity and take the closure in the resulting 
projective space. Then we have to throw away the components
at infinity of the boundary, and for the other components we 
remove their intersection with the hyperplane at infinity. This 
preserves the pure codimension one property.

\subsection{Stability}\label{kempfstab}

Following Kempf \cite{MR506989}, a non-zero
vector $v\in V$ is said to be {\it $G$-stable} if the orbit $G\cdot v$ is closed.
We then also say that $[v]\in \BP V$ is $G$-stable.
If $V=S^dW$ for $\tdim W>3$, $d>3$, and $v\in V$ is generic, then by \cite{MR0396847} its stabilizer
in $SL(W)$ is finite, and by \cite[II 4.3.D, Th. 6 p. 142]{MR768181}, this implies that $v$
is stable  with respect to the $SL(W)$-action.

Kempf's criterion  \cite[Cor. 5.1]{MR506989}   states that if
$G$ does not contain a non-trivial central one-parameter subgroup, and the stabilizer $G([v])$
is not contained in any proper parabolic subgroup of $G$, then $v$ is $G$-stable.
We will apply Kempf's criterion to the determinant in \S 5.2 and to the permanent in \S 5.5.

If $v$ is $G$-stable, then of course $\BC[G\cdot v]=\BC[\ol{G\cdot v}]$.  The former is an  intrinsic
object with the above representation-theoretic description, while the latter is the
  quotient of the space of all polynomials on $V$ by those vanishing on $G\cdot v$.

\subsection{$GL(W)$ v.s.\ $SL(W)$ orbits}\label{slvsglsubsect}

Let $V$ be a $GL(W)$-module and let $v\in V$ be nonzero.
Suppose that the homotheties in $GL(W)$ act non-trivially on $v$.
Then the orbit $GL(W)\cdot v$ is never stable, as it contains
the origin in its closure.

Assume  that $v$ is $SL(W)$-stable, so $\BC[SL(W)\cdot v]=\BC[\ol{SL(W)\cdot v}$]
can be described using~\eqref{LRdecompeqn}.
Unfortunately the ring $\BC[SL(W)\cdot v]$ is not graded.
However $\ol{GL(W)\cdot v}$ is a cone over $SL(W)\cdot v$ with vertex
the origin. The coordinate ring of $\ol{GL(W)\cdot v}$ is equipped with a grading because
 $\ol{GL(W)\cdot v}$ is invariant under rescaling, so any polynomial vanishing on it must also have each of
 its homogeneous components vanishing on it separately. In fact this coordinate ring is the image of a surjective
 map $Sym(V^*)=\BC[V]\ra\BC[\ol{GL(W)\cdot v}]$, given by restriction of polynomial functions,
and this map respects the grading.

  Consider the restriction map
 $\BC[\ol{GL(W)\cdot v}]_\delta\ra \BC[SL(W)\cdot v]$.
 It  is injective for all $\delta$ because a homogeneous polynomial vanishing on
an affine variety vanishes
on the cone over it.
On the other hand,
because  $SL(W)\cdot v$ is a closed subvariety of $\ol{GL(W)\cdot v}$,
restriction of functions yields a surjective map
$\BC[\ol{GL(W)\cdot v}]\ra\BC[SL(W)\cdot v]$.
Both   $\BC[\ol{GL(W)\cdot v}]_\delta$, $\BC[SL(W)\cdot v]$
are $SL(W)$-modules (as  $\ol{GL(W)\cdot v}$ is also an  $SL(W)$-variety), and
the map between them is an $SL(W)$-module map  because the $SL(W)$-action on functions commutes with restriction.

Summing over all $\delta$  yields a surjective $SL(W)$-module map
$$\bigoplus_\delta\BC[\ol{GL(W)\cdot v}]_\delta\ra \BC[SL(W)\cdot v],
$$
that is injective in each degree $\delta$. We have the following consequence observed in \cite{MS2}:

\begin{proposition}\label{slglprop} Let $V$ be a $GL(W)$-module and let $v\in V$
be $SL(W)$-stable.  An irreducible $SL(W)$-module appears in
 $\BC[SL(W)\cdot v]$  iff it  appears in $\BC[\ol{GL(W)\cdot v}]_\delta$ for
 some $\delta$.
\end{proposition}

In contrast to  the case of $SL(W)$,
if an irreducible module occurring in $\BC[GL(W)\cdot v]$ also occurs in
$\BC[\ol{GL(W)\cdot v }]\subset Sym(V^*)$, we can
recover the degree it appears in.
Consider the case $V=S^dW$, then a $GL(W)$-module $S_{\pi}W$ can only occur in $\BC[\ol{GL(W)\cdot v}]$
if $|\pi|=\delta d$ for some~$\delta$ and in that case it can only appear in
$\BC[\ol{GL(W)\cdot v}]_\delta$ 
(see Example \ref{genex} below).

\subsection{Partial stability and an application}\label{MSgobs}

Let $V$ be a $GL(W)$-module.
Let $v,w\in V$ be $SL(W)$-stable points. Equation (\ref{LRdecompeqn}) and Proposition
\ref{slglprop} imply the following observation:
 $w\notin \ol{GL(W)\cdot v}$ (equivalently $\ol{GL(W)\cdot w}\not\subset \ol{GL(W)\cdot v}$)   if there
is an $SL(W)$-module that contains a $SL(W)(w)$-invariant that does not contain a $SL(W)(v)$-invariant.
As discussed below, $\tdet_n$ is $SL(W)$-stable, and while $\ell^{n-m}\tperm_m$ is not $SL(W)$-stable, it
is what is called  {\it partially stable} in \cite{MS2}, which allows one to attempt to search for such modules
as we now describe.

\begin{definition}\cite{MS2}\label{msdef2}
Let $G$ be a reductive group and let $V$ be a  $G$-module.
Let $P=KU$ be a Levi decomposition of a parabolic subgroup~$P$ of $G$.
Let $R$ be a reductive subgroup of $K$.
We say that $[v]\in \BP V$ is {\it $(R,P)$-stable} if it satisfies the two conditions
\begin{enumerate}
\item $U\subset G([v])\subset P$.
\item $v$ is stable under the restricted action of $R$, that is $R\cdot v$ is closed.
\end{enumerate}
\end{definition}

\begin{example} If $x\in S^dW'$ is a generic element and $W'\subsetneq W$ is  a linear inclusion, then
$x$~is not $SL(W)$-stable, but it is $(SL(W'),P)$ stable for $P$ the parabolic
subgroup of $SL(W)$ fixing the subspace $W'\subset W$. This follows from \S 4.3, assuming $\dim W'>3$
and $d>3$.
\end{example}

\begin{example}\label{usefulex}
Let $W=A\op A'\op B$, $A=E\ot F\simeq Mat_{m\times m}$, $\tdim A'=1$, and $G=GL(W)$.
Let $\ell\in A'$ such that $\ell\neq 0$.
It follows from \S\ref{kempfstab} that
$\ell^{n-m}\tperm_m\in S^nMat_{n\times n}^*$ is $(R,P)$-stable for $R=SL(A)$ and $P$
the parabolic subgroup of~$G$ preserving  $A\op A'$, whose Levi factor is
$K=(GL(A\op A')\times GL(B))$.
\end{example}

The point of partial stability is that, since the point $v$ is assumed to be $R$-stable,
the problem of determining the multiplicities of the irreducible
modules $V_\nu (R)$ in
$\BC[\overline{R^.v}]$ is
reduced to the problem of determining the dimension of   $V_{\nu}(R)^{R(v)}$. In the case $R=K$,   these
are also  the multiplicities of the corresponding irreducible representations in the coordinate ring
$\BC[\overline {G^.v}]$.

We will now state a central result of \cite{MS2} (Theorem \ref{secondth} below)  in the special case
that will be applied to $\ell^{n-m}\tperm_m$.
We first need to recall the classical
Pieri formula (see, e.g.,  \cite{weyman}, Proposition 2.3.1 for a proof):

\begin{proposition}\label{spiopdecomp}
For $\tdim A'=1$, one has the $GL(A)\times GL(A')$-module decomposition
$$S_{\pi}(A\op A')=\bigoplus_{\pi\mapsto\pi'}S_{\pi'}A\ot S^{|\pi|-|\pi'|}A',$$
where the notation $\pi\mapsto\pi'$ means that
$\pi_1\ge\pi'_1\ge \pi_2\ge\pi'_2\ge\cdots \ge 0.$
\end{proposition}

\begin{theorem}\label{speccasefirstth}
Let $W=A\op A'\op B$, $\tdim A=\aaa$, $\tdim A'=1$,
$z\in S^{d-s}A$, $\ell\in A'\setminus\{0\}$. Assume $z$ is $SL(A)$-stable.
Write $v=\ell^{s}z$. Set $R=SL(A)$, and take
 $P$ to be  the parabolic of $GL(W)$ preserving $A\op A'$,
so $K=  GL(A\op A')\times GL(B)$, and $z$ is $(R,P)$-stable.

\begin{enumerate}

\item
A module $S_{\nu}W^*$ occurs in $\BC[\ol{GL(W)\cdot v}]_\delta$
iff $S_{\nu}(A\op A')^*$ occurs in $\BC[\ol{GL(A\op A')\cdot v}]_\delta$.
There is then a partition $\nu'$  such that $\nu\mapsto \nu'$
and
$V_{\bl(\nu')}(SL(A))\subset \BC[\ol{SL(A)\cdot [v]}]_\delta$.

\item
Conversely, if $V_{\l}(SL(A))\subset \BC[\ol{SL(A)\cdot [v]}]_\delta$, then
there exist  partitions $\pi, \pi'$ such that $S_{\pi}W^*\subset \BC[\ol{GL(W)\cdot [v]}]_\delta$,
$\pi\mapsto\pi'$ and $\bl(\pi')=\l$.
\item A module $V_{\l}(SL(A))$ occurs in $\BC[\ol{SL(A)\cdot [v]}]$ iff  it occurs in $\BC[SL(A)\cdot v]$.
\end{enumerate}
\end{theorem}


This is a special case of  Theorem \ref{firstth}. It establishes a connection between
$\BC[\ol{GL(W)\cdot v}]$, which we are primarily interested in but we cannot compute,
and $\BC[SL(A)\cdot v]$, which in principle can be described using (\ref{LRdecompeqn}).

We will specialize Theorem \ref{speccasefirstth} to the case $z=\tperm_m$ and study
the precise conditions to have an $SL(A)$-module in $\BC[SL(A)\cdot \tperm_m]$
and the corresponding $GL(W)$-modules in $\BC[\ol{GL(W)\cdot[\ell^{n-m}\tperm_m]}]$.
These conditions are expressed in terms of  certain special Kronecker coefficients,
and we discuss those Kronecker coefficients in \S\ref{kronsect}.

\section{Examples}\label{examplesect}
We study several examples of orbit closures in spaces of polynomials
leading up to the cases of interest, namely
$ \ol{GL_{n^2}\cdot\tdet_n} $, $ \ol{GL_{n^2}\cdot \ell^{n-m}\tperm_m}$,
 $SL_{n^2}\cdot \tdet_n$ and $SL_{m^2}\cdot \ell^{n-m}\tperm_m$.
We also study the coordinate rings of the  orbits
 $GL_{n^2}\cdot \tdet_n$ and $GL_{n^2}\cdot \ell^{n-m}\tperm_m$. For these to be useful,
one must deal with an extension problem,
but the advantage is that their coordinate rings come equipped with a grading
which, when one passes to the closure, indexes the degree.

\subsection{Example:}\label{genex}
Let $W=\BC^n$ and $x\in S^dW$ generic.
We   describe the module structure of
$\BC[GL(W)\cdot x]$ and $\BC[SL(W)\cdot x]$ using \eqref{LRdecompeqn}.
If $x\in S^d W$ is generic and $d,n>3$, then $GL(W)(x) = \{\la Id : \la^d=1  \}\simeq \BZ_d$,
hence  $GL(W)\cdot  x\simeq GL(W)/\BZ_d$,
where $\BZ_d$ acts as multiplication by the $d$-th roots of unity, see \cite{MR0396847}.
(Note that if $x\in S^dW$ is any element,
 $\BZ_d\subset GL(W)(x)$, and thus the calculation here will
be useful for   other cases.)

We   determine the $\BZ_d$-invariants in $GL(W)$-modules. Since
$S_{\pi}W$ is a submodule of $W^{\otimes |\pi|}$, $\o\in \BZ_d$ acts
on $S_{\pi}W\ot (\tdet W)^{-s}$ by the scalar $\o^{|\pi|-ns}$.
By \eqref{LRdecompeqn}, we conclude the following equality of $GL(W)$-modules:
$$
\BC[GL(W)\cdot  x]=  \bigoplus_{(\pi,s)\ \mid\  d\mid |\pi|-ns}(S_{\pi}W^*)^{\op \tdim S_{\pi}W}\ot (\tdet W^*)^{-s}.
$$

Note that $S^{\d}(S^dW^*)$ does not contain any negative powers
of the determinant, so when we pass to
 $\BC[\ol{GL(W)\cdot  x}]=\op_{\d}S^{\d}(S^dW^*)/I_{\d}(\ol{GL(W)\cdot  x})$ we must
loose all terms with $s>0$, i.e.,
we have the inclusion of $GL(W)$-modules
$$
 \BC[\ol{GL(W)\cdot  x}]\subseteq\bigoplus_{\pi\ \mid \ d \mid   |\pi|}(S_{\pi}W^*)^{\op \tdim S_{\pi}W}.
$$
In general there are far fewer  modules and multiplicities in  $S^{\d}(S^{d}W)$  \peterchange
than on the right hand side of the same  degree, 
which illustrates the limitation of this information.
\peterchange
The above inclusion respects degree in the graded module $\BC[\ol{GL(W)\cdot  x}]$:
\be\label{genexgrading}\BC[\ol{GL(W)\cdot  x}]_\delta\subseteq
\bigoplus_{\pi\ \mid \    |\pi|=\delta d}(S_{\pi}W^*)^{\op \tdim S_{\pi}W}.
\ene
This property still holds for any $x\in S^dW$, proving the assertion in the last paragraph
of \S\ref{slvsglsubsect}.

Regarding $SL(W)$, note that $SL(W)(x)=GL(W)(x)\cap SL(W)=\BZ_c$, where $c=gcd(d,n)$.
Thus \eqref{LRdecompeqn} implies,
\be\label{slwgen}
\BC[\ol{SL(W)\cdot x}]=\BC[SL(W)\cdot x]=
\bigoplus_{ \l\in \Lambda^{+}_{SL(W)}\ \mid\  c||\bpi(\l)|  }(V_{\l}^*)^{\op \tdim  V_{\l}}.
\ene

\subsection{First Main Example: $GL(W)\cdot \tdet_n\subset S^nW$}\label{ex:det}
\label{se:1stmainex}

Write $W=E\ot F$, with $E=F=\BC^n$.
The subgroup
$H_0:=\{g\ot h \mid g\in SL(E), h\in SL(F) \}$ of $GL(E\ot F)$
is obtained as the image of $SL(E)\times SL(F)$
under the morphism $(g,h) \mapsto g \ot h$.
The kernel of this morphism equals \peterchange
$\{(\e I,\e^{-1} I) \mid \e^n =1\}$,
which is isomorphic to the group $\mu_n$ of $n$th roots of unity,
so that
$H_0\simeq (SL(E)\times SL(F))/\mu_n$.

Consider the involution $\tau\in GL(E\ot F)$ defined by
$\tau(e \ot f) = f \ot e$ (this makes sense since $E=F$).
We note that
$\tau (g \ot h) \tau = h \ot g$,
so $\tau $ acts nontrivially on $H_0$ by conjugation.
Hence the group
$H := H_0 \langle \tau\rangle \simeq H_0 \rtimes \BZ_2$
is a nontrivial semidirect product.

Frobenius  \cite{Frobdet} showed that the stabilizer of $\tdet_n$ in
$\GL(W)$ equals the group $H$:
\be
 GL(W)(\tdet_n) = H \simeq \big(SL(E)\times SL(F)\big)/\mu_n \rtimes \BZ_2.
\ene
(See \cite{MR2072621} for indications of modern proofs.)
We note that if we interpret $W$ as the space of $n\times n$ matrices~$M$, then
the first factor acts as
$M\mapsto gM h^t$, with $g\in SL(E)$, $h\in SL(F)$,
and $\tau$ acts by transposition $M\mapsto M^t$.

As observed in \cite[Thm.~4.1]{MS1}, $H=GL(W)(\tdet_n)$ is not contained in any proper parabolic subgroup,
so $[\tdet_n]$ is $SL(W)$-stable by Kempf's criterion, see \S\ref{kempfstab}.

Our next goal is to analyze the space $S_\pi(E\ot F)^H$ of $H$-invariants. 
For this, we note that the 
Schur module $S_\mu E$ associated with a partition $\mu\vdash d$
can be characterized as
$S_\mu E = \Hom_{\FS_d}( [\mu], E^{\ot d})$,
where $[\mu]$ denotes the irreducible representation of
the symmetric group $\FS_d$ associated with~$\mu$.
Consider  the vector space 
$K^\pi_{\mu\nu} := \Hom_{\FS_d}\big([\pi],[\mu]\ot [\nu]\big)$
defined for partitions $\mu,\nu,\pi\vdash d$. 
Its dimension 
$k_{\pi\mu\nu} := \dim \Hom_{\FS_d}([\pi],[\mu]\ot [\nu])$
is called the 
 {\it Kronecker coefficient} associated with the partitions $\pi,\mu,\nu$.
The coefficient  $k_{\pi\mu\nu}$  equals the multiplicity of $[\pi]$ in
the tensor product $[\mu]\ot [\nu]$ of 
representations of $\FS_d$. 
We refer to \S 8, and in particular \S\ref{rectsect} for remarks on special Kronecker coefficients.

The canonical linear map
$$
  S_\mu E \ot S_\nu F  \ot K^\pi_{\mu\nu}  \to S_\pi (E\ot F),
 \a \ot \b \ot \g \mapsto (\a\ot\b)\circ \g 
$$
is $GL(E)\times \GL(F)$-equivariant 
(with the trivial action of this group on $K^\pi_{\mu\nu} $). 
Schur-Weyl duality~\cite{FH} tells us that the induced canonical map
\begin{equation}\label{eq:SWmorph}
 \bigoplus_{\mu,\nu\vdash_m d}  S_\mu E \ot S_\nu F \ot K^\pi_{\mu\nu} \ \to\ S_\pi (E\ot F)
\end{equation}
is an isomorphism.
Briefly, \peterchange
the splitting of the Schur module $S_{\pi}(E\ot F)$
with respect to the morphism
$GL(E)\times GL(F) \to \GL(E\otimes F), (g,h)\mapsto g\ot h$
is given by
\begin{equation}\label{eq:defKC}
S_{\pi}(E\ot F)=\op_{\mu,\nu} (S_{\mu}E\ot S_{\nu}F)^{\op k_{\pi\mu\nu}}.
\end{equation}

The action of $\tau\in GL(E\ot E)$ determines \peterchange 
an involution of  $S_{\pi}(E\ot E)$ 
(recall $E=F$). 
We need to understand the corresponding action on the left-hand side 
of~\eqref{eq:SWmorph}.
For this, we note that the isomorphism
$[\mu] \ot [\nu] \to [\nu] \ot [\mu]$ 
resulting from exchanging the factors 
defines a linear map 
$\s^\pi_{\mu\nu}\colon K^\pi_{\mu\nu} \to  K^\pi_{\nu\mu}$
such that  
$\s^\pi_{\nu\mu}\s^\pi_{\mu\nu} = id$. 
It is straightforward to verify that 
\begin{equation}\label{le:tau-action}
 \tau\cdot ( (\a\ot \b) \circ \g)  = (\b \ot \a) \circ \s^\pi_{\mu\nu}(\g) 
\end{equation}
for $\a\in S_\mu E$, $\b \in S_\nu E$, and $\g\in K^\pi_{\mu\nu}$. 
In the case $\mu=\nu$, we get a linear involution 
$\s^\pi_{\mu\mu}$ of $K^\pi_{\mu\mu}$.
 The  subspace \peterchange
of invariants in $K^\pi_{\mu\mu}$ under this involution
can be identified with 
$\Hom_{\FS_d}([\pi], \mathrm{Sym}^2 [\mu])$.
We define the corresponding 
{\em symmetric Kronecker coefficient} as \peterchange
\begin{equation}\label{eq:SKC}
 sk^\pi_{\mu\mu} : = 
\dim \Hom_{\FS_d}([\pi], \mathrm{Sym}^2 [\mu]) .
\end{equation}
So $sk^\pi_{\mu\mu}$ equals the multiplicity of $[\pi]$ in 
the symmetric square $\mathrm{Sym}^2 [\mu]$. \peterchange
Note that 
$sk^\pi_{\mu\mu}  \le k_{\pi\mu\mu}$
and the inequality may be strict.
We refer to~\cite{MR1603309} for some examples. 

The symmetric Kronecker coefficients for rectangular partitions 
$\d^n=(\d\hd \d)$ ($\d$ appears $n$ times) show up in the description of the 
irreducible representions occuring in the coordinate ring of the $GL(W)$-orbit of 
the determinant.

\begin{proposition}\label{peterrefx}
\begin{align}
\BC[GL(W)\cdot  \tdet_n]
&=\bigoplus_{\d\geq 0}\bigoplus_{\pi \, \mid \, |\pi|=n\d }
 (S_{\pi}W^*)^{\op sk^\pi_{\d^n\d^n}}. \label{eq:uno}\\
\BC[\ol{GL(W)\cdot  \tdet_n}]_{\d}
&\subseteq \bigoplus_{ \pi \, \mid \, |\pi|=n\d } (S_{\pi}W^*)^{\op  sk^\pi_{\d^n\d^n}}. \label{eq:due}
\\
\BC[SL(W)\cdot  \tdet_n]
&=\BC[\ol{SL(W)\cdot  \tdet_n}]=\bigoplus_{\l\in \Lambda^+_{SL(W)}}
 ( V_{\l}^*)^{\op sk^{\bpi(\l)}_{\d^n\d^n}},
 \label{eq:tre}
\  \ \d=|\bpi(\l)|/n.
\end{align}
\end{proposition}

\begin{proof}
The multiplicity of $S_{\pi}W^*$ in $\BC[GL(W)\cdot  \tdet_n]$ equals
$\dim S_{\pi}(W)^{H}$ by Equation~\eqref{LRdecompeqn}.
Suppose that $|\pi|= \delta n $ for some~$\delta$. 
Equation~\eqref{eq:SWmorph} implies that 
$$
 (S_{\pi}(E\ot F))^{H_0} =
(S_{\pi}(E\ot F))^{SL(E)\times SL(F)}
 =  S_{\d^n}E \ot S_{\d^n} F \ot K^\pi_{\d^n,\d^n} 
 \simeq  K^\pi_{\d^n,\d^n}  .
$$
For this we used that 
$S_\mu(E)^{SL(E)} = 0$ 
unless $\mu =(\d^n)$ in which case 
$S_\mu(E)^{SL(E)} = \BC$. 
By~\eqref{le:tau-action}  the action of the involution~$\tau$ 
corresponds to the action of $\s^\pi_{\d^n\d^n}$ on $ K^\pi_{\d^n,\d^n} $.
Therefore, 
$\dim (S_\pi (E\ot F))^H = sk^\pi_{\d^n\d^n}$
by the definition of symmetric Kronecker coefficients, 
Moreover, if $n$ does not divide $|\pi|$, then $(S_{\pi}(E\ot F))^{H_0}=0$.
This completes the proof of~\eqref{eq:uno}.

Equation~\eqref{eq:due} is now immediate as
$\BC[\ol{GL(W)\cdot  \tdet_n}]_{\d} \subseteq \BC[GL(W)\cdot  \tdet_n]_{\d}$.
Equation~\eqref{eq:tre} follows from the proof of Equation~\eqref{eq:uno}.
\end{proof}

\subsection{Example:} Suppose $W=A\op B$, with $x\in S^dA$ generic.
Here and below let $\aaa=\tdim A$ and $\bbb=\tdim B>0$. Assume $d,\aaa >3$.
The stabilizer $GL(W)(x)$ of $x$ in $GL(W)$ is of the form
$$
GL(W)(x)=\left\{ \begin{pmatrix} \o Id & *\\ 0&*\end{pmatrix}
\mid \o^d=1\right\}
$$
where the upper $*$ is an arbitrary $\aaa\times \bbb$ matrix, and the lower $*$ is
an arbitrary $\bbb\times \bbb$ invertible matrix.
Since there is no
control over  the lower right hand block matrix in $GL(W)(x)$, an irreducible
$GL(W)$-module $S_{\pi}W\otimes (\det W)^{\ot k}$ can contain non-trivial
invariants only if $k=0$, and then these invariants must be contained in
$S_{\pi}A\subset S_\pi W$. Since $GL(W)(x)$ acts on $S_{\pi}A$ by
homotheties, we conclude that
$$
\BC[GL(W)\cdot  x]=  \bigoplus_{\pi\ \mid \ d\mid |\pi|,\ \ell(\pi)\leq\aaa }(S_{\pi}W^*)^{\op \tdim S_{\pi}A}.
$$
In particular,   all modules $S_{\pi}W^*$ with $d||\pi|$ and
$\ell(\pi)\leq\aaa$ do occur. The elimination of modules with more than $\aaa$ parts
is due to our variety being contained in a {\it subspace variety} (defined in \S\ref{wexa} below), consistent
with Proposition \ref{coordbootprop}.

For comparison with what follows, we record the following immediate consequence for all $\delta$:
\be\label{ocinfo}
\BC[\ol{GL(W)\cdot  x}]_\delta\subseteq   \bigoplus_{\pi\ \mid \
|\pi|=d\delta,\ \ell(\pi)\leq\aaa }(S_{\pi}W^*)^{\op \tdim S_{\pi}A}.
\ene

Since $x$ is not $SL(W)$-stable, we instead use the $(SL(A),P_{\aaa})$-partial stability of $x$ to obtain
further information. Namely take
$R=SL(A)$,   $K=GL(A)\times GL(B)$, and $P_{\aaa}$ the parabolic preserving  $A$.
From \eqref{slwgen} we have a description of $\BC[SL(A)\cdot x]$ in terms of $c=gcd(d,\aaa)$.
By Theorem \ref{speccasefirstth}, for each dominant integral weight~$\lambda$
of $SL(A)$ such that $c$ divides $|\bpi(\l)|$, some~$\pi$ with $\bl(\pi)=\l$ must occur in
$\BC[\ol{GL(W)\cdot x}]$,  and by \eqref{genexgrading}
it occurs in $\BC[\ol{GL(W)\cdot x}]_{|\pi|/d}$.

\subsection{Example:}\label{subcaseap}

Suppose $W=A\op A'\op B$ and  $x=z\ell^{s}\in S^dW$,
where $z\in S^{d-s}A$ is generic, and
$\tdim A'=1$, $\ell\in A'\setminus\{0\}$.
Assume $d-s,\aaa >3$.
It is straightforward to show that, with respect to bases adapted to the splitting $W=A\op A'\op B$,
$$
GL(W)(x)=\left\{ \begin{pmatrix}   \psi Id & 0& *\\ 0&\eta &*\\ 0 & 0 &*\end{pmatrix}
\mid  \ \eta^{s} \psi^{d-s}=1\right\}.
$$
Working as above, we first observe that the $GL(W)(x)$-invariants in $S_\pi W$ must be contained in
$S_\pi (A\oplus A')$. By the Pieri formula \ref{spiopdecomp}, this is the
sum of the $S_{\pi'}A\otimes S^{|\pi|-|\pi'|}A'$, for $\pi\mapsto\pi'$. The
action of $GL(W)(x)$ on such a factor is by multiplication with
$\psi^{|\pi'|}\eta^{|\pi|-|\pi'|}$, hence the conditions for invariance that $|\pi'|=\delta(d-s)$
and $|\pi|=\delta d$ for some $\delta$.
We conclude that
$$
\BC [GL(W)\cdot  x]=\bigoplus_{\delta\geq 0}\bigoplus_{\substack{|\pi|= \delta d, \; |\pi'|=\delta (d-s),
\\ \pi\mapsto\pi'}}
(S_{\pi}W^*)^{\op \tdim S_{\pi'} A },
$$
\begin{equation}\label{xconclu}
\BC [\ol{GL(W)\cdot  x}]_\delta\subseteq\bigoplus_{\substack{|\pi|=\delta d, \;|\pi'|=\delta (d-s), \\ \pi\mapsto\pi'}}
(S_{\pi}W^*)^{\op \tdim S_{\pi'} A }.
\end{equation}

\medskip

The point  $x$ is not $SL(W)$-stable, but is $SL(A)$-stable, and thus $(R,P)$-stable for $(R,P)=(SL(A),P_{\aaa+1})$.
Theorem \ref{speccasefirstth} applied to this case says that if $S_{\pi}W^*\subset \BC[\ol{GL(W)\cdot x}]_\delta$ then
$S_{\pi}(A\op A')^*\subset\BC[\ol{GL(A\op A')\cdot x}]_\delta$ and there exists $\pi'$     such that $\pi\mapsto\pi'$
and $V_{\l (\pi')}(SL(A)) \subset \BC[SL(A)\cdot x]$. Moreover, by \eqref{slwgen} the latter condition is equivalent
to the condition that $c=gcd(d-s,\aaa)$ divides $|\pi'|$.

\subsection{Example:}\label{permcase}

Suppose $W=Mat_{m\times m}$ and $x=\tperm_m$.
We write $W=E\ot F$, with $E=F=\BC^m$.
Let $T_E$ denote the maximal torus of diagonal matrices in $SL(E)$.
Its normalizer $N_E$ is the semidirect product of $T_E$ and the Weyl group~$\cW_E$
of permutation matrices in $GL(E)$. 
Similarly, let $T_F$ denote the maximal torus of $SL(F)$ and
$N_F=T_F\rtimes \cW_F$ its normalizer.
If we denote by $N_0$ the image of $N_E\times N_F$ under
$GL(E)\times GL(F) \to GL(E\ot F), (g,h)\mapsto g\ot h$, then
$N_0\simeq \big( N_E\times N_F\big)/\mu_m$,
where $\mu_m$ denotes the group of $m$th roots of unity.
Recall from \S\ref{ex:det} the involution $\tau\in GL(E\ot F)$
and consider the subgroup
$N := N_0 \langle \tau\rangle \simeq N_0 \rtimes \BZ_2$.

By \cite{MR0137729}, for $m>2$,
the stabilizer of $\tperm_m\in S^m(E\ot F)$ equals
\be\label{permstabeqn}
 GL(W)(\tperm_m) =  N \simeq \big (N_E\times N_F\big)/\mu_m \rtimes \BZ_2 .
\ene
(It is  stated in~\cite{MS1} that the stabilizer is   found in
\cite{MR504978}, although this is not correct.
A shorter proof of \eqref{permstabeqn} is given in \cite{botta:67}.)

In \cite[Theorem 4.7]{MS1} it is observed that $SL(W)(\tperm_m)$ is not contained in
any proper parabolic subgroup of $SL(W)$,
so $\tperm_m$  is $SL(W)$-stable by Kempf's criterion, see~\S\ref{kempfstab}.

Consider the Schur module $S_\mu E$ corresponding to
a partition $\mu \vdash_m \d m$.
Then the zero weight space $(S_{\mu}E)_0 := (S_\mu E)^{T_E}$
of $S_\mu E$ with respect to the $SL(E)$-action is nonzero.
The group $\cW_E$ acts on $(S_{\mu}E)_0$
and we shall denote by
$p_\mu :=  \dim (S_{\mu}E)_0^{\cW_E}$
the dimension of the space of its $\cW_E$-invariants.
In fact, Corollary~\ref{cor93} stated later on,
identifies~$p_\mu$ as the following {\it plethysm coefficient}:
$$
 p_\mu = \tmult(S_{\mu} E, S^m(S^\delta E)) .
$$

\begin{definition}\label{sigpermm}
Define \peterchange 
$\Sigma_{\tperm_m}\subset \Lambda^+_{GL_{m^2}}$ to be
the set of partitions $\pi$ such that:
\begin{enumerate}
\item $|\pi |=\delta m$ some $\delta\in \BN$,
\item there exist $\mu,\nu\vdash_m \d m$ with
$p_\mu p_\nu \ne 0$ and either
\begin{enumerate}
\item $k_{\pi\mu\nu}\neq 0$  if $\mu\ne\nu$ or 
\item $sk^\pi_{\mu\mu}\neq 0$  if $\mu=\nu$.
\end{enumerate}
\end{enumerate}
For $\pi\in  \Sigma_{\tperm_m}$, define \peterchange 
$$
 mult_{\pi}= \frac 12 \sum_{\mu\neq \nu} k_{\pi\mu\nu} p_\mu p_\nu
+ \sum_{\mu} sk^\pi_{\mu\mu}\binom{ p_\mu +1}{2}.
$$
Note that $\tmult_\pi\ge 1$ for $\pi\in  \Sigma_{\tperm_m}$.
Finally let $\Sigma_{\tperm_m}^S=\bpi^{-1}(\Sigma_{\tperm_m})\subset\Lambda^+_{SL_{m^2}}$.
\end{definition}

\begin{proposition}\label{permringpropa}
\begin{align} \label{eq:puno}
 \BC[GL(W)\cdot  \tperm_m]
&=\bigoplus_{\pi\in  \Sigma_{\tperm_m}}(S_{\pi}W^*)^{\op mult_{\pi}}.\\ \label{eq:pdue}
 \BC[\ol{GL(W)\cdot \tperm_m}]_\delta
&\subseteq
\bigoplus_{\substack{\pi\in  \Sigma_{\tperm_m}, \\|\pi|=\delta m}}(S_{\pi}W^*)^{\op mult_{\pi}}.\\
\label{eq:ptre}
 \BC[\ol{SL(W)\cdot  \tperm_m}]
&=\BC[SL(W)\cdot  \tperm_m]=
\bigoplus_{\l\in \Sigma_{\tperm_m}^S
}(V_{\l}^*)^{\op mult_{\bpi(\l)}}.
\end{align}
\end{proposition}

\begin{proof}
By Equation~\eqref{LRdecompeqn} we need to show that
$\dim S_{\pi}(W)^{GL(W)(\tperm_m)} = mult_{\pi}$.
From~\eqref{eq:SWmorph} 
we obtain, using $(S_\mu E)^{T_E} = (S_\mu E)_0$, that
$$
 (S_{\pi}(E\ot F))^{ T_E \times  T_F } =
 \bigoplus_{\mu,\nu} (S_{\mu}E)_0 \ot (S_{\nu}F)_0 \ot K^\pi_{\mu\nu}  ,
$$
which implies, using $N_E=T_E\ltimes\cW_E$, that \peterchange
$$
 (S_{\pi}(E\ot F))^{N_E\times N_F} =
 \bigoplus_{\mu,\nu} (S_{\mu}E)_0^{\cW_E} \ot (S_{\nu}F)_0^{\cW_F} \ot K^\pi_{\mu\nu} .
$$
For proving Equation~\eqref{eq:puno}, 
it remains to show that $mult_\pi$ equals the dimension of the space
of $\tau$-invariants of $(S_{\pi}(E\ot F))^{N_E\times N_F}$.
Put $X_\mu := (S_{\mu}E)_0^{\cW_E}$ to simplify notation. 
Equation~\eqref{le:tau-action} implies that for $\mu\ne\nu$, 
the space of $\tau$-invariants
$$
 \Big(X_\mu \ot X_\nu \ot K^\pi_{\mu\nu}\ \oplus\  X_\nu \ot X_\mu \ot K^\pi_{\nu\mu} \Big)^\tau
$$
projects bijectively onto 
$X_\mu \ot X_\nu \ot K^\pi_{\mu\nu}$. 
Moreover, 
$$
 \Big(X_\mu \ot X_\mu \ot K^\pi_{\mu\mu}\Big)^\tau = 
  \mathrm{Sym}^2(X_\mu) \ot (K^\pi_{\mu\mu})^\tau .
$$  
Taking into account $p_\mu =  \dim (S_{\mu}E)_0^{\cW_E}$, 
it follows that 
$ (S_{\pi}(E\ot F))^{N_0\langle\tau\rangle} = mult_\pi$
as claimed in~\eqref{eq:puno}, 

Equation~\eqref{eq:pdue} is now immediate as
$\BC[\ol{GL(W)\cdot  \tperm_m}]_{\d} \subseteq \BC[GL(W)\cdot  \tperm_m]_{\d}$.
Equation~\eqref{eq:ptre} follows from the proof of Equation~\eqref{eq:puno}.
\end{proof}


\subsection{Second Main Example}\label{fmainex}

Let $W=A\op A'\op B$, $A=E\ot F\simeq Mat_{m\times m}$, $\tdim A'=1$, $\tdim W=n^2$,
and $x=\ell^{n-m}\tperm_m$, $\ell\in A'$.
With respect to bases adapted to the splitting $W=A\op A'\op B$,
\begin{equation}\label{gpermeqn}
GL(W)(x)=\left\{ \begin{pmatrix}   \xi  GL(W)(\tperm_m)  & 0& *\\ 0&\eta &*\\ 0 & 0 &*\end{pmatrix}
\mid \eta^{n-m} \xi^{m}=1\right\}.
\end{equation}

\begin{definition}\label{sigmapermmn}
For $n > m$, define $\Sigma^n_{\tperm_m}\subset \Lambda^+_{GL_{n^2}}$ to be
the set of partitions $\pi$ such that:
\begin{enumerate}
\item $  |\pi |=\delta n$ some $\delta\in \BN$,
\item there exists $\pi'\in \Sigma_{\tperm_m}$, such that $|\pi' |=\delta m$ and $\pi\mapsto \pi'$.
\end{enumerate}
Moreover, for $\pi\in\Sigma^n_{\tperm_m}$ we set
$$\tmult^n_\pi = \sum_{\substack{\pi'\in \Sigma_{\tperm_m},\; \pi\mapsto \pi'\\ n|\pi' |=m|\pi|}}\tmult_{\pi'}.$$
\end{definition}

  Proposition \ref{permringpropa} and   Example    \ref{subcaseap} show:

\begin{proposition}\label{permringprop}
\begin{align*}
\BC[GL(W)\cdot  \ell^{n-m}\tperm_m]
&=\bigoplus_{\pi\in  \Sigma_{\tperm_m}^{n}}(S_{\pi}W^*)^{\op mult^n_{\pi}},
\end{align*}
\begin{equation}\label{cheapperm}
\BC[\ol{GL(W)\cdot  \ell^{n-m}\tperm_m}]_\delta\subseteq
\bigoplus_{\substack{\pi\in  \Sigma_{\tperm_m}^{n},\\ |\pi|=n\delta}}(S_{\pi}W^*)^{\op mult^n_{\pi}}.
\end{equation}
\end{proposition}


\medskip

Since $SL(W)\cdot \ell^{n-m}\tperm_m$ is not stable, we consider
$R=SL(A)$ as in \S\ref{subcaseap}. (We could have augmented $R$ by the semi-simple part of the stabilizer
of $\ell^{n-m}\tperm_m$ but this would not yield any new information.)


From Theorem~\ref{speccasefirstth} we deduce the following result.

\begin{proposition}\label{permringpropx}
$\ell^{n-m}\tperm_m$ is $(SL(A),P_{m^2+1})$-partially stable.
Thus for all $\l \in \Sigma^S_{\tperm_m}$, there exist
partitions $\pi,\pi'$ such that $\bl(\pi')=\l$, $\pi\mapsto\pi'$, and
$S_{\pi}W^*\subset \BC[\ol{GL(W)\cdot \ell^{n-m}\tperm_m}]$.
\end{proposition}

Since  in Proposition \ref{permringpropx}  we have
  no information about which degree a module appears in, for each
$\lambda$ there are an infinite number of $\pi$'s that could be associated
 to it. Thus Proposition~\ref{permringpropx} may be difficult to utilize in practice.

Proposition \ref{permringpropx} combined with Theorem  \ref{speccasefirstth}
gives an explicit description of the Kronecker problem that results
from \cite{MS2} regarding the permanent.

\section{\lq\lq Inheritance\rq\rq\ theorems and desingularizations}

 In \S \ref{pstabvalueb} we   explain the approach to determine the coordinate ring of an orbit closure outlined in \cite{MS2}.
In \S\ref{basicthmsubsect}
we review the geometric  method for desingularizing $G$-varieties by collapsing a homogeneous vector
bundle. We then, in \S\ref{wexa}, \S\ref{wexb}   give two examples of auxiliary varieties
that can be studied with such desingularizations and are useful for the problems at hand.
We discuss how this perspective can be used to recover Theorems \ref{firstth} and \ref{secondth} from \cite{MS2}  and to obtain
  further   information that might be useful.

\subsection{Inheritance  theorems appearing in \cite{MS2}}\label{pstabvalueb}


Let $R\subseteq K\subset G$ be as in Definition \ref{msdef2}. We can choose a maximal torus of
$G$ in such a way that its intersections with $R$ and $K$ are maximal tori in these subgroups.
This allows one to identify weights accordingly, i.e., it induces restriction maps $\Lambda_G\simeq\Lambda_K
\ra\Lambda_R$, and we impose that $\Lambda_G^+\ra\Lambda_K^+\ra\Lambda_R^+$.

\begin{definition}\label{liedef}
We say that   $\nu\in \Lambda^+_G$ {\it lies over}
$\mu\in \Lambda^+_R$ at $v$ and degree $\delta$ if
\begin{enumerate}
\item $V_\mu (R)^*$ and $V_{\nu}(K)^*$ occur in $\BC[\overline{R^. [v]}]_\delta$ and
$\BC[\overline{K^. [v]}]_\delta$ respectively,
\item $V_\mu (R)^*$ occurs in $V_{\nu}(K)^*$ considered as an $R$-module.
\end{enumerate}
We say that a dominant weight $\nu$ of $G$ lies over a dominant weight
$\mu$ of $R$ at $v$ if this happens for some $\d>0$.
\end{definition}

\begin{example}\label{lieexample} (Example \ref{usefulex} cont'd)
 Let $W=A\op A'\op B$, $\tdim A=\aaa$, $\tdim A'=1$, $v=\ell^{s}z$ with $\ell\in A'$,
$z\in S^{d-s}A$ such that $z$ is $SL(A)$-stable, so setting $R=SL(A)$, $P$ the parabolic
subgroup of $GL(W)$ preserving
$A\op A'$,  $v$ is $(R,P)$-stable. Suppose that
a  weight in  $\Lambda^+_{GL(W)}$ defined by some partition $\pi$, lies over
$\l\in \Lambda^+_{SL(A)}$.

First, that $S_\pi W^*$ be contained in $\BC [GL(W)\cdot v]$ requires that $\ell(\pi)\leq\aaa+1$
(which will also be justified in \S\ref{wexa} by the fact that $\ol{GL(W)\cdot [v]}$ lies in the
subspace variety $Sub_{\aaa+1}(W)$).
Second, the condition that $V_\l(SL(A))$ be contained in the restriction of $S_\pi (A\oplus A')^*$
requires that $\pi\mapsto\pi'$ for some partition $\pi'$ such that $\ell(\pi')\le \aaa$ and
$\l (\pi')=\l$.
Finally  we need $V_{\l}(SL(A))$ to occur in $\BC[\ol{SL(A) \cdot[v]}]_\delta$.
Theorem \ref{firstth} below describes when this occurs for some $\delta$.
\end{example}

\begin{definition}\cite{MS2}\label{msdef1}
Let $H\subset G$ be a subgroup. We say that a $G$-module $M$ is {\it $H$-admissible} if it
contains a non-zero $H$-invariant. We let $M^H\subset M$ denote the subspace of $H$-invariants.
Note that an irreducible $G$-module is $H$-admissible iff it appears in $\BC [G/H]$.
\end{definition}

\begin{theorem}[\cite{MS2}, Theorem 8.1] \label{firstth}
Let $[v]\in \BP V$ be $(R,P)$-stable. Then the representation $V_\lambda(G)$ occurs in the coordinate ring
$\BC[ \overline{G\cdot [v]}]$   only if $\lambda$ lies over   some $R(v)$-admissible dominant weight $\mu$ of $R$.
Conversely, for every $R(v)$-admissible dominant weight $\mu$ of $R$, $\BC[ \overline{G\cdot[v]}]$ contains
$V_\lambda (G)$ for some dominant weight $\lambda$ of $G$ lying over   $\mu$ at $v$.
\end{theorem}

Theorem \ref{firstth} is a consequence of  the following more precise result.

\begin{theorem}[\cite{MS2}, Theorem 8.2] \label{secondth}
Let $[v]\in\BP V $ be $(R,P)$-stable. Let $P=KU$ be a Levi decomposition of $P$. Then:
\begin{enumerate}
\item
A $K$-module $V_\l (K)^*$ occurs in $\BC[\overline {K^.[v]}]$ only if $\l$ is also dominant for $G$, and for all $\delta$
$$
\tmult(V_\lambda (G)^*,\BC[\overline {G^.[v]}]_\delta)=\tmult(V_\lambda (K)^*,\BC[\overline {K^.[v]}]_\delta).
$$
\item
There are inequalities
$$
\tmult(V_\lambda (G)^*, H^0 (\overline{G^.[v]}, \mathcal O_{\overline{G^.[v]}} (\delta)))\le
\tmult(V_\lambda (K)^*, H^0 (\overline{K^.[v]}, \mathcal O_{\overline{K^.[v]}} (\delta))).
$$
\item A $K$-module $V_\l (K)^*$ can occur in $\BC[\overline {K^.[v]}]_\delta$   only if
$\l\in\Lambda_G^+$ lies over some $\mu\in\Lambda_R^+$ at $v$ and degree $\delta$.
Conversely, for each $R$-module $V_\mu (R)^*$ occurring in $\BC[\overline{R^.[v]}]_\delta$, there exists a
$G$-dominant weight $\l$ lying over $\mu$  at $v$ and degree $\delta$.
\item An $R$-module $V_\mu (R)^*$ occurs in $\BC[\overline{R^.[v]}]$ if and only if it is $R(v)$-admissible.
\end{enumerate}
\end{theorem}

\noindent {\it Idea of proof}.
These statements relate the coordinate rings of the projective orbit closures $\ol{G\cdot [v]}$,
$\ol{K\cdot [v]}$, $\ol{R\cdot [v]}$, and of the affine (closed) orbit $R\cdot v$.

In order to prove (1), one observes that the
surjective map
$$
\BC[\ol{G\cdot[v]}] \twoheadrightarrow \BC[\ol{K\cdot[v]}]
$$
is not only a $K$-module map, but also a $P$-module map where the $P$-module structure on the right-hand side is obtained
by extending the action of $K$ by the trivial action of $U$. (This relies on the assumption that
$G([v])$ contains $U$.) Any copy of $V_\l (G)^*$  in some $\BC[\ol{G\cdot[v]}]_\delta$ maps to a $P$-module $N$
which is
non-zero, because if all polynomials in a $G$-module vanish on $[v]$, they must also vanish on $\ol{G\cdot [v]}$.
Dualizing, since the action of $U$ on $N$ is trivial, one gets an
injection   $N^*\rightarrow V_\l (G)^U$,
whose image is  the irreducible module $V_\l (K)$. In particular $N^*$ is irreducible.
This  implies (1), and its variant (2) is proved in a similar way.

In order to prove (3), one simply observes that the surjection
$$
\BC[\ol{K\cdot[v]}] \twoheadrightarrow \BC[\ol{R\cdot[v]}]
$$
is non-zero on any irreducible component of $\BC[\ol{K\cdot[v]}]_\delta$, by the same argument as above.
So any such $V_\l (K)^*$ contributes to $\BC[\ol{R\cdot[v]}]_\delta$ by some $V_\mu (R)^*$ for weights
$\mu$ over which $\l$ lies. Conversely any component of $\BC[\ol{R\cdot[v]}]_\delta$ is obtained that way
since the restriction map is surjective.

Finally, (4) is a consequence of the fact that $R\cdot v$ is contained in the
cone over $\ol{R\cdot [v]}$. Since they are both closed in $V$, this yields a surjection
$$
\BC[\ol{R\cdot[v]}] \twoheadrightarrow \BC[R\cdot v]
$$
and the same argument as for the proof of Proposition \ref{slglprop} shows that
both sides involve the same  irreducible modules. \qed

\medskip
We emphasize that (4) gives no information of the degree in which a given irreducible module
may occur in $\BC[\overline{R^.[v]}]$.

In this paper we do not discuss (2), whose failure to be an equality is related with the failure of the
cone over $\ol{K\cdot[v]}$ to be normal, hence to the type of singularity that
occurs at the origin.

\smallskip
There is a connection between the notion of $(R, P)$-stability and the collapsing method
that we discuss in the next subsections. From the latter perspective it is easy to deduce the relationship
between $\BC[\ol{K\cdot[v]}]$ and $\BC[\ol{G\cdot[v]}]$, although the relationship between these
and $\BC[\ol{R\cdot[v]}]$ is more subtle.
It is possible to write   alternative proofs of Theorems \ref{firstth}, \ref{secondth}   using the collapsing
set-up.

The desingularization method could be useful for several reasons.
First, it allows one to calculate the multiplicity of an
irreducible $G$-module $V_\lambda (G)$
  in each graded component of the coordinate ring of an orbit closure.
One could detect that one orbit is not in the closure of the other by
comparing these multiplicities.
Second, it gives information about the multiplicative structure of the
coordinate ring. If an orbit $\cO_1$ is in the closure of an orbit $\cO_2$
then the coordinate ring $\BC [{\overline \cO}_1]$ is a quotient of $\BC [{\overline
\cO}_2]$ so every polynomial relation   in $\BC [{\overline \cO}_2]$
still holds in $\BC [{\overline \cO}_1]$.
Finally, desingularization gives information about the singularities of an
orbit closure, which are   important geometric invariants.

\subsection{The collapsing method and its connection with partial stability}\label{basicthmsubsect}

The following statement can be extracted from \cite[Chapter 5]{weyman}:

\begin{theorem}\label{weymanthm}
 Let $Y\subset \BP V$ be a
projective variety. Suppose there is a
projective variety $\cB$ and a vector bundle
$q: E\ra \cB$ that is a subbundle of a trivial bundle
$\underline V \ra \cB$ with fiber $V$,
such that the image of the map $\BP E\ra \BP V$ is $Y$ and  $\BP E\ra   Y$ is a desingularization
of $Y$.  Write $\eta=E^*$ and $\xi=(\underline V/E)^*$.

If the sheaf cohomology groups
$H^i(\cB,S^\delta\eta)$ are all zero for $i>0$ and $\delta>0$,
and if the linear maps
$H^0(\cB,S^\delta\eta)\ot V^*\ra H^0(\cB,S^{\delta+1}\eta)$
are surjective for all $\delta\geq 0$, then
\begin{enumerate}
\item
$\hat Y$ is normal, with rational singularities.

\item The coordinate ring $\BC[\hat Y]$ satisfies
  $\BC[\hat Y]_\delta\simeq H^0(\cB,S^\delta\eta)$.



\item If moreover $Y$ is a $G$-variety
and the desingularization is $G$-equivariant,
then the identifications above are as $G$-modules.
\end{enumerate}
\end{theorem}


Notations as above, assume that $v\in V$ is $(R,P)$-stable. Let $W=\langle K\cdot v\rangle$ be the
smallest $K$-submodule of $V$ containing $  v$. Since $  v$ is stabilized by $U$, and $U$ is normalized
by $K$, $W$ is a $P$-submodule of $V$ with a trivial $U$-action.
Consider the diagram
$$\begin{CD}
E_W:=G\times_P W @>p>> G/P \\
@VV{q}V  \\
Z_W \subset V.
\end{CD}$$
where $E_W$ is a vector bundle over $G/P$ with fiber $W$, and $Z_W:=q(E_W)= \overline{G^.W}=G^.W$.
The coordinate ring of $Z_W$ is a subring of    $H^0 (G/P, Sym (E_W^* ))$.
In the case when $q$ is a desingularization  (i.e.,  when $q$ is birational),
  $H^0 (G/P, Sym (E_W^* ))$ is the normalization of the coordinate ring of $Z_W$.

The orbit closure $\overline{ K\cdot v}$ is a $K$-stable subset of $W$, and the  method
of \cite{weyman} reduces the calculation of the $G$-module structure of   $\BC[\overline{G\cdot v}]$ to the calculation of $K$-module structure of   $\BC[\overline{K\cdot v}]$.

\subsection{The subspace variety}\label{wexa}

Let $W$ be a vector space and for $\aaa <\tdim W$ define
$$Sub _{\aaa}(S^dW) =\lbrace f\in S^d W\ \ \ |\ \ \exists\ W'\subset W,\  {\rm dim}(W')={\aaa} , f\in S^d W'\subset S^dW \rbrace .$$
  $Sub_{\aaa} (S^dW)$ is a closed subvariety of $S^dW$ which has a
  natural desingularization given by the total space
of a vector bundle over the Grassmannian  $Gr(\aaa,W)$, namely $GL(W)\times_PS^d\BC^{\aaa}=S^d\cS$, where $\cS\ra Gr(\aaa, W)$ is the
tautological subspace bundle over the Grassmannian. In other words, the total space of $S^d\cS$ is
$$ \lbrace (f, W')\in S^dW\times G  (\aaa , W)\ \ |\ \ f\in S^dW' \rbrace .$$

Using Theorem \ref{weymanthm}  one may determine
 the generators of the   ideal $I(Sub _{\aaa}(S^dW) )$ as follows.
For $\phi\in S^dW$ and $\d<d$, consider the \lq\lq flattening \rq\rq\ $\phi_{\d,d-\d}: S^\d W^*\ra S^{d-\d}W$
via the inclusion $S^dW\subset S^{\d}W\otimes S^{d-\d}W$.

\begin{proposition} (\cite{weyman}, \S 7.2)\label{wey72}
\begin{enumerate}
\item The   ideal $I(Sub_{\aaa} (S^dW))$ is the span of all submodules $S_\pi W^*$ in $Sym (S^dW^* )$ for which $\ell(\pi)>\aaa$.
\item   $I(Sub_{\aaa} (S^dW))$ is generated by
$\La{\aaa+1}W^*\ot \La{\aaa+1}(S^{d-1}W^*)$, which may be considered as  the span of the
$(\aaa +1)\times (\aaa +1)$ minors of $\phi_{1,d-1}$.
\item $Sub _{\aaa}(S^dW)$ is normal, Cohen-Macaulay and it has rational singularities.
\end{enumerate}
\end{proposition}

  Proposition \ref{wey72} implies:

\begin{proposition}\label{coordbootprop}Let $W'\subset W$ be a subspace of dimension $\bbb$ and let $f\in S^dW'$.
Assume that the coordinate ring of the orbit closure    $\ol{GL(W')^.f}\subset S^dW'$ has the $GL(W')$-decomposition
$$\BC [\overline{GL(W' )^.f}]=\bigoplus_{\pi ,\ell(\pi )\le \bbb} (S_\pi W'^*)^{\op m(\pi )}.$$
Then the coordinate ring of the orbit closure   $\ol{GL(W)^.f}\subset S^dW$ has the $GL(W)$-decomposition
$$\BC[\overline{GL(W)^.f}] =\bigoplus_{\pi ,\ell(\pi)\le \bbb} (S_\pi W^*)^{\op m(\pi )} .$$
\end{proposition}

\begin{proof}We actually prove a more precise statement about the two ideals.
First note that $\ol{GL(W)\cdot f}\subset Sub_{\bbb}(S^dW)$ so for all partitions $\pi$
with $\ell(\pi)>\bbb$, and $S_{\pi}W^*\subset Sym(S^dW^*)$,  $S_{\pi}W^*\subset I(\ol{GL(W)\cdot f})$.
So henceforth we consider only partitions $\pi$ with $\ell(\pi)\leq \bbb$.

We will show that $S_{\pi}W^*\subset I(\overline{GL(W)^.f})$
iff $S_{\pi}W'^*\subset  I(\overline{GL(W')^.f})$
for any partition $\pi$ with $\ell(\pi)\leq \bbb$.
Assume $|\pi|=d\d$ (this must be the case for
$S_{\pi}W^*$ to appear in 
$S^{\d}(S^dW^*)$) and $\ell(\pi)\leq \bbb$. Some highest weight vector of
$S_{\pi}W^*\subset S^{\d}(S^dW^*)$ lies in $S^{\d}(S^dW'^*)$.  That it vanishes on $GL(W)\cdot f$ implies it vanishes
on $GL(W')\cdot f$  because if we choose a splitting $W=W'\op W''$ and write
$h\in S^dW$ as $h= h_1+h_2 $ with $h_1\in S^dW'$, $h_2|_{S^dW'}=0$,
 given $p\in S^{\d}(S^d{W'}^*)$, we have   $p(h)=p(h_1)$,     and
$h\in GL(W)\cdot f$ iff $h_1\in GL(W')\cdot f$.
 Finally, an irreducible $G$-module
vanishes on a $G$-variety iff any highest weight vector vanishes on the
variety.
\end{proof}

\begin{remark} The statements above are the special cases of the first part of Theorem \ref{secondth} in the case when $W=W'\op W''$, and
$G=GL(W)$, $K=GL(W')\times GL(W'')$.
\end{remark}

Applying Proposition \ref{coordbootprop} to $x=\ell^{n-m}\tperm_m\in S^n\BC^{m^2+1}=W'\subset W=\BC^{n^2}$
  reduces the problem of determining $\BC[\ol{GL(W)\cdot \ell^{n-m}\tperm_m}]$
to determining $\BC[\ol{GL(W')\cdot \ell^{n-m}\tperm_m}]$.

\subsection{Polynomials divisible by a linear form}\label{wexb}\label{linfactorcoord}

Another ingredient in the collapsing approach to Theorem \ref{firstth} is investigating
a variety of polynomials divisible by a power of a linear form.

\begin{problem}\label{jerzysug}
\label{prob}
Let $W'\subset W$ be a   subspace of codimension one. Let $\ell\in W\setminus W'$.
Let $g\in S^{d-s}W'$. Take $f=\ell^s g\in S^dW$. Compare the decompositions of the coordinate rings
of the orbit closures $\ol{GL(W' )^.g}$ and $\ol{GL(W)^.f}$.
\end{problem}

A solution to   Problem \ref{jerzysug} would   reduce the investigation of the orbit of $\ell^{n-m}\tperm_m$ to the orbit closure of the permanent itself.

\medskip

Consider the subvariety
$$F_s (S^dW)=\lbrace f\in S^dW\ |\ f=\ell^sg\ \ {\rm for\ some} \ \ell\in W, g\in S^{d-s}W\rbrace .$$

The variety $F_s (S^d W)$ arises naturally in the GCT program because
one is interested in the coordinate ring of    $\ol{GL(W)\cdot \ell^{n-m}perm_m}$ which is contained in $F_{n-m} (S^n W)$.
 The description of the normalization of $F_s (S^d W)$ should be useful   because the coordinate ring
of $F_s(S^d W)$ is a subring in the coordinate ring of its normalization. This normalization is
best understood via a collapsing as follows.

The closed subvariety $F_s (S^dW)$ has a desingularization of the form in Theorem \ref{weymanthm}
with $G/P=\BP W$, i.e.,  $P$ is the parabolic subgroup of $GL_n$   stabilizing a subspace of dimension one, and
the bundle $\eta=S^s{\mathcal S}^*\otimes S^{d-s}W^*$,     where $\mathcal S=\cO_{\BP W}(-1)$ is the tautological
subbundle  over   $\BP W$. The higher cohomology of $Sym (\eta )$ vanishes.
  Theorem \ref{weymanthm} implies that the normalization of the coordinate ring of $F_s (S^dW)$  has the decomposition
$$Nor(\BC [ {F_s (S^dW)}])_e=  S^{es}W^*\otimes S^e (S^{d-s}W^* ).$$
This decomposition implies that  $\BC [ {F_s (S^dW)}]$ is non-normal  because
$\BC[ {F_s (S^dW)}]_1=S^dW^*$ and $Nor(\BC [ {F_s (S^dW)}])_1=S^sW^*\ot S^{d-s}W^*$,
but on the other hand  if $X$  is a normal, affine variety and $f:Y\rightarrow X$
 is a desingularization, then
$H^0 (Y, {\mathcal O}_Y )=H^0 (X, {\mathcal O}_X)$.
   Thus, to determine $\BC[F_s(S^dW)]$
one would   need to deal with  the non-normality of $F_s (S^dW)$. However,  in the
situation of the proof of Theorem \ref{firstth}  it is possible to partially avoid such issues.

\section{Orbits and their closures}\label{orbitclosuresect}

\subsection{Comparing $\ol{GL_{n^2}\cdot\tdet_n}$ and $\fgl_{n^2}.\tdet_n$}\label{gpvsmatsect}

In this section we compare the orbit closure $\ol{GL(W)\cdot \tdet_n}$
with the orbit $GL(W)\cdot \tdet_n$ and the set $End(W)\cdot \det_n$.
The reasons for the first comparison have been discussed already -
the second comparison could be useful for helping to understand the
first, and it is also important because Valiant's conjecture is related
to $End(W)\cdot \det_n$.

In our July 2009 preprint  we asked if  one had the equality
$\ol{GL(W)\cdot  \tdet_n}=End(W)\cdot\tdet_n$. Since then, it
has been shown that the equality fails, see \cite[Prop. 3.5.1]{LMRdet}.
\peterchange


\medskip

A method to construct polynomials belonging to
$\ol{GL(W)\cdot  \tdet_n}$ but not to $End(W)\cdot \tdet_n$
is proposed in  \cite[pp. 508-510]{MS1}. The idea is to start from a weighted graph
$G$ with $n$ (ordered) vertices, with $n$ even.  Consider its skew-adjacency matrix $M_G$, the skew-symmetric matrix
whose $(i,j)$-entry with $i<j$ is a variable $y_{ij}$ if there is an edge between the vertices $i$ and $j$,
and zero otherwise. More generally, define $M_G(t)$ as before but replacing $y_{ij}$ by $t^{w_{ij}}y_{ij}$,
where $w_{ij}\in \BZ_{>0}$ denotes the weight of the edge $ij$. Then
$$\det (M_G(t))=[\mathrm{Pfaff}\, M_G(t)]^2=t^{2W}h_G(y)+\mathrm{higher\;order\;terms},$$
where $W$ is the minimal weight of a perfect matching of $G$, and $h_G(y)$ is a sum of monomials indexed
by pairs of minimal perfect matchings.  By construction, the polynomial $h_G(y)$ is in $\ol{GL(W)\cdot  \tdet_n}$.
In general $G$ has a unique minimal perfect matching, so $h_G(y)$ is just a monomial which belongs
to $End(W)\cdot\tdet_n$. It is conjectured in   \cite[\S 4.2]{MS1}   that there   exist  pathological weighted graphs $G$ such
that $h_G(y)$ does not have a small size formula  and does not belong to $End(W)\cdot \tdet_n$.

\subsection{Towards understanding $\ol{GL(W)\cdot \tdet_n}\subset S^nW$}

In order to better understand   the coordinate ring of $\ol{GL(W)\cdot  \tdet_n}$, it will be
important to answer the following question:

\begin{question} What are the irreducible components of the boundary
of $\ol{GL(W)\cdot  \tdet_n}$? Are they $GL(W)$-orbit closures?
\end{question}


In principle   $\ol{GL(W)\cdot  \tdet_n}$ can be analyzed as follows. The action of $GL(W)$ or
$End(W)$ on $\tdet_n$ defines a rational map
$$\psi_n : \PP (End(W))\dashrightarrow \PP (S^nW^*)$$
given by $[u]\mapsto [  \tdet_n\circ u]$.
Its  indeterminacy locus   $I(\psi_n)$ is, set theoretically, given by the set of $u$ such that $\tdet(u.X)=0$ for
all $X\in W=Mat_{n\times n}$. Thus
$$I(\psi_n)=\{u\in End(W) \mid \;\; Im(u)\subset  Det_n\},$$
where $Det_n\subset W$ denotes the hypersurface of non-invertible matrices. Since $Im(u)$ is
a vector space, this relates the problem of understanding $\psi_n$ to that of linear subspaces in
the determinantal hypersurface
$\BP (Det_n) \subset \BP (End(W))$, which has already received some attention
(see e.g. \cite{MR954659}.)

By Hironaka's theorems \cite{MR0199184} one can resolve the indeterminacy locus of $\psi_n$ by a sequence
of smooth blow-up's, and $\ol{GL(W)\cdot  \tdet_n}$ can then be obtained  as the image of the resolved map.
Completely resolving  the indeterminacies will probably be too difficult, but this approach should
help to answer the preceeding questions.

\subsection{Remarks on the extension problem}\label{extensionsect}

Let $G$ be reductive, let $V$ be an irreducible $G$-module and let $v\in V$.
Consider the closure $\ol{G\cdot v}$ of the $G$-orbit $G\cdot v\simeq G/G(v)$.
Then the boundary $\ol{G\cdot v}\backslash G\cdot v$ has finitely many components
$H_1,\ldots H_N$ of codimension at least one in $\ol{G\cdot v}$. If $G$ is connected,
each of these components is a $G$-variety.
Moreover, if $G(v)$ is reductive, then all $H_i$ have codimension one, cf.~\S\ref{staborbitsect}.\peterchange


\smallskip

\begin{example}The most classical example of all for the extension problem is:
$\BC^*\subset \BC$:
$\BC[\BC^*]=\BC[z,z\inv]$  and
  $\BC[\ol{\BC^*}]=\BC[\BC]=\BC[z]$.
Here we can take $G=\BC^*$, $v=1$.
\end{example}

\smallskip



Consider the case where  the singular locus of $\ol{G\cdot v}$ has codimension at least two.
Then the generic point of each codimension one $H_i$ is a smooth point of $\ol{G\cdot v}$, so that
$H_i$ can be defined around that point by a regular function $h_i$,
uniquely defined up to an invertible function. This allows one to define a
valuation $\nu_i$ on $\CC[ {G\cdot v}]$, giving the order of the
pole of a rational function along $H_i$: each regular
function $f$ on $ {G\cdot v}$, considered as a rational function of
$\ol{G\cdot v}$, can be uniquely written at the generic point of $H_i$
as $f=gh_i^{\nu_i(f)}$, where  $g$ is regular and invertible, and
$\nu_i(f)\in\mathbb{Z}$. The valuation $\nu_i$ is $G$-invariant if
$H_i$ is. Since a regular function on $\ol{G\cdot v}$ has no poles, we have
$$
\BC[\ol{G\cdot v}]\subset \{f\in \BC[G\cdot v] \mid \; \forall i\ \nu_i(f)\ge 0\}.
$$
If moreover $\ol{G\cdot v}$ is normal, then equality holds: if $f\in \CC[G\cdot v]$,
is such that $\nu_i(f)\ge 0$ for all~$i$, then $f$ is regular at the generic
point of any codimension one boundary component of $\ol{G\cdot v}$, hence
outside a subset of codimension at least two -- hence everywhere (see, e.g., \cite{eisenbud}, Corollary 11.4).
(Earlier,  Kostant (\cite{MR0158024} , Proposition 9, p 351) showed that if the
boundary of $\ol{G\cdot v}$ has codimension at least two in $\ol{G\cdot v}$,   and
 $\ol{G\cdot v}$ is normal, then $\BC[\ol{G\cdot v}]=\BC[{G\cdot v}]$ ).

In July 2009 we wrote that we expected
this normality condition and the codimension two
singularities condition to fail in our cases.
Since then,  Kumar~\cite{kumar:10} proved
that neither the orbit of the determinant nor of the permanent are normal varieties.
Nevertheless, the analysis of codimension one boundary
components of the orbit ${G\cdot v}$ should be a first step towards
the determination of $\CC[\ol{G\cdot v}]$.
We also point out that the boundary of the orbits of the permanent and determinant 
are of pure codimension one, as 
their stabilizers  are reductive, cf.~\S\ref{staborbitsect}.\peterchange


Another instance of an extension problem was the problem essentially solved by Demazure
for   $B$-orbits in $G/B$, where
$G$ is semi-simple and $B\subset G$ a Borel subgroup. Here the orbits,
which are  Schubert cells,
are just affine spaces (and thus have very simple coordinate rings) and
the closures are Schubert varieties. For a precise, more general statement,
and references, see \cite[Theorem 8.2.2]{MR1923198}.
This result relies on the normality of the Schubert varieties, which,
as remarked above,
fails for the orbit closures of interest here.

Finally, we remark that a recent work~\cite{buic:10} tries to apply the
GCT-approach to to problem of proving lower bounds on tensor rank.
One of the main outcomes of this work is that by looking at SL-obstructions only
trivial lower bounds can be shown.

\section{Kronecker coefficients}\label{kronsect}

We have seen that we need to understand the Kronecker coefficients $k_{\d^{n},\d^{n},\pi}$ in order to
understand $\BC[GL(W)\cdot \tdet_n]$. Similarly, in order to understand
$\BC[GL(W)\cdot \ell^{n-m}\tperm_m]$ we need to understand Kronecker coefficients
$k_{\pi\mu\nu}$ where $S_{\mu}\BC^m$ and  $S_{\nu}\BC^m$ are contained in some plethysm
$S^m(S^k\BC^m)$.    We first give general facts about
computing Kronecker coefficients which tell us the multiplicities
of certain modules in the coordinate rings we are interested in. Since keeping track of the multiplicities
in the cases at hand appears to be hopeless, one could try to solve the simpler question of
non-vanishing of Kronecker coefficients (i.e., that a certain module appears at all), so
we next discuss conditions where one can
determine if Kronecker coefficients   are non-zero. Finally in the last two subsections we
specialize to the types of Kronecker coefficients arising in the study of
$\BC[\tdet_n]$   and $\BC[\ell^{n-m}\tperm_m]$.

\subsection{General facts}  A general reference for this section is \cite[\S I.7]{macdonald}.
Let $\pi, \mu, \nu$ be three partitions of a  number $n$. The Kronecker coefficient
$k_{\pi\mu\nu}$ is the dimension of the space of $\FS_n$-invariants in
$[\pi]\otimes [\mu]\otimes [\nu]$, where recall that $[\pi]$ is the irreducible $\FS_n$-module associated
to $\pi$.  In particular $k_{\pi\mu\nu}$ is {\it symmetric} with respect to $\pi, \mu, \nu$.   Since the irreducible
complex representations of $\FS_n$ are all defined over $\BQ$, $k_{\pi\mu\nu}$ is also
the multiplicity of $[\pi]$ inside the tensor product $[\mu]\otimes [\nu]$.

Write $\pi = (n-|\bar\pi|,\bar\pi)$. Then $k_{\pi\mu\nu}$ only depends on
the triple $(\bar\pi, \bar\mu, \bar\nu)$ when $n$ is sufficiently large, cf. \cite{mur:38}.
A more precise statement was obtained in \cite{MR1243152}. It implies that if  $k_{\pi\mu\nu}\ne 0$, then
$|\bar\pi|\le |\bar\mu|+|\bar\nu|$. Moreover, in case of equality, the Kronecker coefficient
can be identified with a Littlewood-Richardson coefficient:
$$k_{\pi\mu\nu}=c^{\bar\pi}_{\bar\mu , \bar\nu}.$$

\subsubsection*{Relation with characters}
Kronecker coefficients can be computed from the characters of the irreducible representations
of $\FS_n$. Let $\chi_\pi$ denote the character of $[\pi]$. Then (see \cite[p. 115]{macdonald})
\begin{equation}\label{char}
k_{\pi\mu\nu}=\frac{1}{n!}\sum_{w\in\FS_n}\chi_\pi(w)\chi_\mu(w)\chi_\nu(w).
\end{equation}
The characters of $\FS_n$ can be computed in many ways.  Following the Frobenius character formula,
they appear as
coefficients of the expansion of Newton symmetric functions $p_\mu$ in terms of Schur functions $s_{\pi}$:
$$p_\mu = \sum_{\pi}\chi^\mu_\pi s_\pi.$$
Here $\chi^\mu_\pi$ denotes the value of the character $\chi_\pi$ on any  permutation of cycle type
$\mu$. Another   formula for $\chi^\mu_\pi$ is given by the Murnaghan-Nakayama rule, which involves
a certain type of tableaux $T$ of shape $\pi$ and weight $\mu$ (that is, numbered in such a way that
each integer $i$ appears $\mu_i$ times). Call $T$ a {\it multiribbon tableau}  if it is numbered non-decreasingly
on each row and column, in such a way that for each $i$, the set of boxes numbered $i$ forms a ribbon
(a connected set containing no two-by-two square). Then
$$\chi^\mu_\pi = \sum_T(-1)^{h(T)},$$
where the sum is over all multiribbon tableaux $T$ of shape $\pi$ and weight $\mu$,
and $h(T)$ is the sum of the heights of the ribbons in $T$ (the height of a ribbon being the
number of rows it occupies, minus one). See e.g. \cite[I.7, Ex.5]{macdonald}.

\subsubsection*{Small length cases}
The symmetric group $\FS_n$ has two one dimensional representations, the trivial representation $[n]$
and the sign representation $[1^n]$. One has
$$[n]\otimes [\pi]=[\pi]\qquad \mathrm{and}\qquad [1^n]\otimes [\pi]=[\pi^*],$$
where $\pi^*$ denotes the conjugate partition of $\pi$. After these two,  the simplest representation
of $\FS_n$ is the vector representation $[n-1,1]$ on $n$-tuples of complex numbers with sum zero.
Its exterior powers $\wedge^p[n-1,1]=[n-p,1^p]$ are   irreducible.
Recently Ballantine and Orellana \cite{MR2264933}
computed the product of $[n-p,p]$ with $[\pi]$
under the condition that $\pi_1\ge 2p-1$ (or $\pi^*_1\ge 2p-1$).

\subsubsection*{Schur-Weyl duality}
There is a close connection between representations of symmetric groups and representations
of general linear groups, called  Schur-Weyl duality \cite{MR1321638}.   Consider the tensor power $U^{\otimes n}$
of a complex vector space $U$. The diagonal action of $GL(U)$ commutes with the permutation
action of $\FS_n$. Schur-Weyl duality is the statement that, as a $GL(U)\times\FS_n$-module,
$$U^{\otimes n} = \bigoplus_{|\pi|=n}S_{\pi}U\otimes [\pi].$$
A straightforward consequence is the already stated fact that the Kronecker coefficient
$k_{\pi\mu\nu}$ can be defined as the multiplicity of $S_{\mu}V\otimes S_{\nu}W$ inside
$S_{\pi}(V\otimes W)$ (at least for $V$ and $W$ of large enough dimension). In particular,
since $[n]$ is the trivial representation, this yields the Cauchy formula
$$S^n(V\otimes W)=\bigoplus_{|\pi|=n}S_{\pi}V\otimes S_{\pi}W.$$
Using the Giambelli formula (which expresses any Schur power in terms of symmetric
powers) and the Cauchy formula, it is easy to express any
Kronecker coefficient in terms of Littlewood-Richardson coefficients.
If $\pi$ has length~$\ell$,
we denote the multiplicity of
$S_{\mu}V$ in $S_{\a_1}V\otc S_{\a_{\ell}}V$
by $c_{\alpha_1,\ldots,\alpha_\ell}^\mu$.
Then
\begin{equation}\label{klr}
k_{\pi\mu\nu}=\sum_{w\in\FS_\ell}{\rm{sgn}}(w)\sum_{\substack{(\alpha_1,\ldots,\alpha_\ell), \\
|\alpha_i|=\pi_i-i+w(i)}}c_{\alpha_1,\ldots,\alpha_\ell}^\mu c_{\alpha_1,\ldots,\alpha_\ell}^\nu.
\end{equation}

\subsection{Non-vanishing of Kronecker coefficients}\

\subsubsection*{The semi-group property}
A rephrasing of the Schur-Weyl duality yields the decomposition
\begin{equation}\label{triple}
Sym(U\otimes V\otimes W)=\bigoplus_{\pi,\mu,\nu}
(S_{\pi}U\otimes S_{\mu}V\otimes S_{\nu}W)^{\oplus k_{\pi\mu\nu}}.
\end{equation}
Using the fact that the highest weight vectors in this algebra form a
finitely generated subalgebra, one can deduce (see \cite{MR2276458}) that:
\begin{itemize}
\item Triples of partitions with non-zero Kronecker coefficients
form a semi-group; that is, if $k_{\pi\mu\nu}\ne 0$ for three partitions $\pi, \mu, \nu$
of some integer $n$, and  $k_{\pi'\mu'\nu'}\ne 0$ for three partitions $\pi', \mu', \nu'$
of $n'$, then
$$k_{\pi+\pi',\mu+\mu',\nu+\nu'}\ne 0.$$
\item   If one restricts to triples of partitions
of length bounded by some integer $\ell$, the corresponding semi-group is finitely generated.
\item
If $k_{\pi\mu\nu}\ne 0$, the normalized partitions
$\tilde\pi=\frac{\pi}{n}, \tilde\mu=\frac{\mu}{n}, \tilde\nu=\frac{\nu}{n}$
verify the entropy relations
\be\label{star12}
H(\tilde\pi)\le H(\tilde\mu)+H(\tilde\nu).
\ene
  Here $H(\tilde\pi)=-\sum_i\tilde\pi_i\log(\tilde\pi_i)$ denotes the {\it Shannon entropy} \cite{MR0026286}.
\end{itemize}

Saturation does not hold for Kronecker coefficients, that is, $k_{N\pi,N\mu,N\nu}\ne 0$ for some $N\ge 2$
does not imply that $k_{\pi,\mu,\nu}\ne 0$. For counter-examples, see \cite{BORpreprint}, whose appendix by
Mulmuley contains several conjectures regarding the saturation property.

\subsubsection*{Linear constraints for  vanishing}\label{linear}
Consider the set $\mathrm{KRON}$ of triples $(\tilde\pi, \tilde\mu, \tilde\nu)$,
where $\pi, \mu, \nu$ are three partitions of $n$ such that $k_{\pi\mu\nu}\ne 0$
and $\tilde\pi$ etc.\ are as above.
Let  $\mathrm{KRON}_\ell$ denote the analogous set with the additional condition that the length
of the three partitions be bounded by $\ell$. One can deduce from the previous remarks that
$\mathrm{KRON}_\ell$ is a rational convex polytope (see e.g. \cite{MR1923785} and \cite{MR2276458}).

What are the equations of the facets of this polytope? A geometric method to produce many such facets
appears in \cite{MR1465785}, in terms
of embeddings
$$\varphi_T : \mathcal{F}(V)\times \mathcal{F}(W)\hookrightarrow \mathcal{F}(V\otimes W).$$
Here $\mathcal{F}(V)$ (resp. $\mathcal{F}(W)$) denotes the variety of full flags in the vector space
$V$ (resp. $W$), of dimension $m$ (resp. $n$).
There is no canonical way to define a flag $H$ in $V\otimes W$ from a flag $F$ in $V$ and a flag
$G$ in $W$. In order to do that, one needs to prescribe what Klyachko calls a {\it cubicle}: a numbering
$T$ of the boxes $(i,j)$ of a rectangle $m\times n$ by integers $\ell_T(i,j)$ running from $1$ to $mn$,
increasingly on each line and column. Then one lets
$$H_k=\varphi_T(F,G)_k=\sum_{\ell_T(i,j)\le k} F_i\otimes G_j.$$
Klyachko \cite{Klyachkopreprint} goes one step further by applying results of \cite{MR1750957}.
To state his result, we need a definition. Consider two non-increasing
sequences $a$ and $b$ of real numbers, of lengths $m$ and $n$, each of  sum zero. Suppose that
the real numbers $a_i+b_j$ are all distinct. Ordering them defines a sequence $a+b$ of length $nm$,
thus a cubicle $T$ and the associated map $\varphi_T$.
Recall that the integral cohomology ring $H^*(\mathcal{F}(V))$ has a natural
basis given by the Schubert classes $\sigma_u$, indexed by permutations $u\in\FS_m$. For any
permutation $w\in\FS_{mn}$, we can therefore decompose the pull-back by $\varphi_T$ of the corresponding
Schubert class as
$$\varphi_T^*\sigma_w=\sum_{\substack{u\in\FS_m \\ v\in\FS_n}}c^w_{uv}(a,b)\sigma_u\otimes\sigma_v.$$
The coefficients $c^w_{uv}(a,b)$ are non-negative integers. Klyachko's statement is the following:

\begin{theorem}\cite{Klyachkopreprint}
Suppose $\ell\ge m,n$. Then $(\tilde\pi, \tilde\mu, \tilde\nu)$ belongs to $\mathrm{KRON}_\ell$
if and only if
$$\sum_i a_i\tilde\pi_{u(i)}+\sum_j b_j\tilde\mu_{v(j)}\ge \sum_k (a+b)_k\tilde\nu_{w(k)}$$
for all non-increasing sequences $a,b$ and for all
$u\in \FS_m, v\in \FS_n, w\in \FS_{mn}$ such that $c^w_{uv}(a,b)\ne 0$.
\end{theorem}

There is a formula for the coefficients $c^w_{uv}(a,b)$ in terms of divided differences operators,
which allows one  to make explicit computations in low dimensions. For example one can recover
the description of $\mathrm{KRON}_3$ given by M. Franz \cite{MR1923785}
as the convex hull of $11$ explicit points. Unfortunately there is no general rule for deciding
whether $c^w_{uv}(a,b)$ is zero or not. Moreover the number of inequalities seems to grow
extremely fast with $\ell$. Redundancy is also an issue. Klyachko conjectures that it is
enough, as for the Horn problem, to consider inequalities for which $c^w_{uv}(a,b)=1$.
Recent advances by N. Ressayre~\cite{resspreprint} allow one, in
principle, to get a complete and irredundant list of facets for  $\mathrm{KRON}_{\ell}$.

In \cite{buci:09}  the set of
$(\tilde{\pi},\tilde{\mu},\tilde{\nu})\in \mathrm{KRON}$ with the additional condition that
$\tilde{\mu},\tilde{\nu}$ are the uniform distributions of length $\ell$ were studied.
The resulting $\tilde{\pi}$ can be {\em any}
probability distribution on $\ell^2$ points
so that the containment in $\mathrm{KRON}$ does not impose any constraint.
This is significant  in view of Proposition~\ref{peterrefx} and shows that
``candidates'' for obstructions are in a sense rare.

\subsection{Case of rectangular partitions}\label{rectsect}

\subsubsection*{Stanley's character formula}
Formula \eqref{char} shows that, in order to compute a Kronecker coefficient of type
$k_{\delta^n,\delta^n,\pi}$, it would be useful to have a nice formula for the character
$\chi_{\delta^n}$. Recall that $\delta^n$ denotes the partition whose diagram is a rectangle $\delta\times n$
(i.e.,  the partition $(\delta\hd \delta)=(\delta^n)$). Such a formula is given by Stanley in
\cite{MR2049555}.
Suppose that $w$ is a permutation in $\FS_{\delta n}$.
Then
$$\chi_{\delta^n}(w) = \frac{(-1)^{\delta n}}{\prod_{i=1}^\delta\prod_{j=1}^n(i+j-1)}
\sum_{uv=w}\delta^{\kappa(u)}(-n)^{\kappa(v)},$$
where $u,v\in\FS_{\delta n}$ and $\kappa(u)$ denotes the number of cycles in $u$.

\subsubsection*{Relations with invariants}
Let $U,V,W$ be vector spaces of dimensions $\ell, n,n$ respectively.
Taking $SL(V)\times SL(W)$-invariants in Formula \eqref{triple} yields
$$A:=Sym(U\otimes V\otimes W)^{SL(V)\times SL(W)}=\bigoplus_{\delta,\pi}
(S_\pi U)^{\oplus k_{\pi,\delta^n,\delta^n}} .$$

For $\ell=2$ it is known that $A\simeq Sym(S^nU)$, \cite[Theorem 17   p.~369]{MR1800533}. Thus for a partition $\pi=(a,b)$ of $\delta n$
in two parts, $k_{\pi,\delta^n,\delta^n}$
is equal to the multiplicity  of $S_\pi U$ in  $S^\delta(S^nU)$. This is given
by {\it Sylvester's formula} (see, e.g., \cite[Theorem 3.3.4]{MR0447428}):
\be\label{sylformula}
k_{(\delta n- b,b), \delta^n, \delta^n}=P(b;\delta\times n)-P(b-1;\delta\times n),
\ene
where $P(b;\delta\times n)$ denotes the number of partitions of size $b$ inside the
rectangle $\delta\times n$.

This also follows directly from formula \eqref{klr},
once we observe that a Littlewood-Richardson coefficient $c_{\alpha, \beta}^{\delta^n}$
is non-zero only if  $\alpha$ and $\beta$ are complementary partitions in the
rectangle $\delta\times n$, and in that case it equals one (this is a straightforward consequence
of the Littlewood-Richardson rule, and a version of Poincar\'e duality for Grassmannians).

The same argument yields a formula for the length three case as follows. Let $\pi=(a,b,c)$ with $a+b+c=\delta n$.
Denote by $ST(a,b;\delta\times n)$ the number of semistandard lattice permutation skew-tableaux whose shape is
of the form $\beta/\alpha$, for $\beta$ a partition of size $\delta n-b$ in the rectangle
$\delta\times n$, and $\alpha$ a partition of size $a$ (see \cite{macdonald} for the terminology). Then
\begin{eqnarray*}
k_{\pi, \delta^n, \delta^n} &= &ST(a,b;\delta\times n)-ST(a,b+1;\delta\times n)+ST(a+1,b+1;\delta\times n)\\
     & &-ST(a+1,b-1;\delta\times n)+ST(a+2,b-1;\delta\times n)-ST(a+2,b;\delta\times n).
\end{eqnarray*}

\medskip
For $n=2$, and $\dim U=4$, the algebra of highest weight vectors in $A$  turns out to be polynomial,
with generators of weight $(2),(22),(222)$ and $(1111)$ \cite{arXiv:0809.3710}. Call a partition even
(respectively odd) if all its parts are even (respectively odd). We deduce:

\begin{prop}
A Kronecker coefficient $k_{\pi,(\delta\delta),(\delta\delta)}$ is non-zero if and only if:
\begin{itemize}
\item either $\pi$ is an even partition of $2\delta$, of length at most four,
\item or $\pi$ is an odd partition of $2\delta$, of length exactly four.
\end{itemize}
In both cases $k_{\pi,(\delta\delta),(\delta\delta)}=1$.
\end{prop}

\subsubsection*{Constraints}
Let $[\pi]$ be a component of $[(\delta^n)]\otimes [(\delta^n)]$.
The entropy relations \eqref{star12} yield
$$H(\tilde\pi)\le 2\log(n).$$
Denote $|\pi|_{\le a}=\pi_1+\cdots+\pi_a$ (and similarly
$|\pi|_{\ge a}$, etc...). Then \cite[Th\'eor\`eme 3.2]{MR1465785} gives 
$$|\pi|_{> ab}\le \delta (n-a)^+ +\delta (n-b)^+$$
where $x^+=x$ if $x$ is positive and zero otherwise.
For example $|\pi|_{\le n}\ge \delta$.

\subsection{A variant of Schur-Weyl duality}

By Schur-Weyl duality, the decomposition of the Schur powers
$S_\pi(V_1\otimes\cdots\otimes V_m)$ into irreducible components, for $|\pi|=\ell$,
is equivalent to the decomposition of tensor products of $m$ irreducible representations
of $\FS_\ell$. What happens if we let $V_1=\cdots =V_m=V$ and replace the tensor product
$V_1\otimes\cdots\otimes V_m$ by the $m$-th symmetric power of $V$?

The following remarkable
theorem is proved in \cite{MR0414794}. Suppose $V$ has dimension $n$, and fix
a basis of $V$. This defines an action of $\FS_n$ on $V$, and on any Schur power $S_\mu V$.
In particular the zero-weight space $(S_\mu V)_0$ is an $\FS_n$-module, non-trivial if and only if
$\mu$ is of size $n\delta$ for some $\delta$.
Here zero-weight must be understood with respect to a maximal torus in~$SL(V)$.

\begin{theorem}\cite{MR0414794}\label{wtzerothm} Let $\tdim V=n$ and let  $\mu$ be a partition of $n\delta$ (so that
$(S_{\mu}V)_0\neq 0$).
Suppose that the decomposition of $(S_\mu V)_0$ into irreducible $\FS_n$-modules is
$$(S_\mu V)_0=\bigoplus_{\pi}[\pi]^{\oplus s_{\mu,\pi}} .$$
Then one has the decomposition of $GL(V)$-modules
$$S_{\pi}(S^\delta V)=\bigoplus_{\mu}(S_\mu V)^{\oplus s_{\mu,\pi}}.$$
\end{theorem}

In particular, for $\delta=1$, i.e., $|\mu|=n$,     $(S_\mu V)_0=[\mu]$.

\begin{corollary}\label{cor93}
Let $\mu$ be a partition of size $n\delta$.
The dimension of the space of $\FS_n$-invariants in the zero weight space $(S_\mu \BC^n)_0$
equals the multiplicity of  $S_\mu \BC^n$   in the plethysm $S^n(S^\delta\BC^n)$.
\end{corollary}

For $\delta=2$, because of the formula \cite[Ex. 6(a), p. 138]{macdonald},
this implies that $(S_\mu V)_0$ contains non-trivial $\FS_n$-invariants
if and only if $\mu$ is even.
For general $\delta$,
conditions for multiplicities  not to vanish have been obtained in \cite{MR1243152} and \cite{MR1651092}.
Recently,
in response to   our paper, \marginpar{$\leftarrow$}
it was shown in~\cite{buci:10} that whenever $\delta$ is even and
all the parts~$\mu_i$ are of even size, then
$S_\mu \BC^n$ occurs $S^n(S^\delta\BC^n)$.
Hence $(S_\mu \BC^n)_0$ contains $\FS_n$-invariants in this case.

Observe that for $n=\dim V=2$, these multiplicities are given by Sylvester's formula \eqref{sylformula}.
This can be generalized as follows. Consider a finite dimensional $GL(V)$-module $M$,
and let $m_\mu(M)$  denote   the multiplicity of the weight $\mu$ in $M$.  Let $N_\pi(M)$
denote
the multiplicity of $S_\pi V$ in the decomposition of $M$ into irreducible components.
Then
\begin{equation}\label{npi}
N_{\pi}(M)=\sum_{w\in\FS_n}sgn(w)m_{w(\pi+\rho)-\rho}(M),
\end{equation}
where $\rho=(n,\ldots ,2,1)$.
Indeed, the Weyl  character formula is equivalent to \eqref{npi} when $M$ is irreducible.
By linearity, it must hold for any $M$. In particular, let $M=S^n(S^\delta V)$. The multiplicity
$m_\mu(M)$ is then equal to the number $p(\mu;n,\delta)$
of ways of writing the monomial $x^\mu$ as a product of $n$ monomials of degree $\delta$.
The multiplicity of $S_{\pi}V$ inside $S^n(S^\delta V)$ is thus
$$N(\pi;n,\delta)=\sum_{w\in\FS_n}sgn(w) p(w(\pi+\rho)-\rho;n,\delta),$$
which generalizes Sylvester's  formula.

\section{Complexity classes}\label{complexityclasssect}

In this section we explain the precise complexity problem studied by the GCT program,
namely $\ol{\VPws}\neq \vnp$,
and place it in the context of Valiant's
algebraic model of NP-comple\-te\-ness \cite{vali:79-3,vali:82}.
In particular, we compare this to the conjecture $\vp\neq \vnp$,
and that the permanent is not a p-projection of the determinant,
the latter being equivalent to the conjecture $\VPws\neq \vnp$.
The conjecture $\vp\neq \vnp$ is an arithmetic analog
of the conjecture $\bold P\neq \bold{NC}$.

All polynomials considered are over $\BC$.
A general reference for this section is \cite{buer:00-3}.

\subsection{Models of arithmetic circuits and complexity}
An {\em arithmetic circuit} is a finite acyclic directed graph with vertices of in-degree~$0$
or $2$ and exactly one vertex of out-degree~$0$. Vertices of in-degree~$0$ are called
{\em inputs} and labeled by a constant in $\BC$ or a variable. The other vertices, of
in-degree~$2$, are labeled by $\times$ or $+$ and called {\it computation gates}.
We define the {\em size} of a circuit as the number of its vertices.
The {\em depth} of the circuit is defined as the maximum length of a directed path
in the underlying graph.
The polynomial computed by a circuit is easily defined by induction.

If the graph underlying the circuit is a directed tree, i.e., all
vertices have out-degree at most~$1$, then we call the circuit an
{\em expression} or {\em formula}.
The notion of {\em weakly-skew} circuits is less restrictive: we require that
for each multiplication gate~$\alpha$, at least one of the
two vertices pointing to $\alpha$ is computed by a separate
subcircuit~$C_\alpha$. Separate means that the edge connecting
$C_\alpha$ to $\alpha$ is the only edge between a vertex of
$C_\alpha$ and the remainder to the circuit.
In short, formulas are circuits where previously computed values cannot be reused,
while in weakly-skew circuits we require that at least one of the two operands
of a multiplication gate is computed just for that gate.
We note that the degree of the polynomial computed by a weakly-skew circuit
is bounded by its size. The motivation for weakly skew-circuits is that they
exactly characterize the determinant, as we explain below.

We define the {\em complexity} $L(f)$ of a polynomial~$f$ over $\BC$ as the minimum
size of an arithmetic circuit computing~$f$. Restricting to weakly-skew circuits
and formulas, respectively, one defines the corresponding complexity notions
 $\wsL(f)$ and $L_e(f)$. Clearly, $L_e(f)\geq\wsL(f)\geq L(f)$.
The quantity $L_e(f)$ is called the {\em formula size} of $f$.
It is an important fact~\cite{bren:74} that $\log L_e(f)$ equals, up to a constant factor,
the minimum depth of an arithmetic circuit computing~$f$.

An algorithm due to Berkowitz~\cite{berk:84}
for computing the determinant
implies $\wsL(\tdet_n) = \Oh(n^5)$.
This algorithm also shows the well-known fact that
$\log (L_e(\det_n)) =\Oh(\log^2 n)$.
The best known upper bound
$L(\per_m) = \cO(m 2^m)$
on the complexity of the permanent
is exponential~\cite{MR0150048}.

The complexity class $\VPe$ is defined as the set of sequences $(f_n)$ of
multivariate polynomials over $\BC$ such that $L_e(f_n)$ is polynomially bounded in $n$.
The set of sequences $(f_n)$ such that $\wsL(f_n)$ is polynomially bounded in $n$
comprises the complexity class $\VPws$.
The class $\VP$ is defined as the the set of sequences $(f_n)$ such that
$L(f_n)$ and $\deg f_n$ are polynomially bounded in~$n$
(it is possible to give a syntactic characterization of $\VP$ in terms
of multiplicatively disjoint circuits~\cite{mapo:04}).
Note that $\VPe\subseteq\VPws\subseteq\VP$.
Since $\wsL(\det_n) = \Oh(n^5)$,
we have $(\det_n)\in\VPws$.
It is a major open question whether $(\det_n)$ is contained in $\VPe$.
This is equivalent to the question whether $\det_n$ can be computed by arithmetic circuits
of depth $\Oh(\log n)$. The best known upper bound is $\Oh(\log^2 n)$,  see  \cite{berk:84}.

\subsection{Completeness}\label{completesect}

A polynomial $f$ is called a {\em projection} of a polynomial $g$ if $f$ can be
obtained from $g$ by substitution of the variables by variables or constants.
A sequence $(f_n)$ is called a {\em p-projection} of a sequence $(g_n)$ if there exists a polynomially
bounded function $t\colon\N\rightarrow\N$ such that $f_n$ is a projection of $g_{t(n)}$ for all~$n$.
We note that each of the previously introduced complexity classes $\mathcal{C}$ is
closed under p-projection, i.e., if $(f_n)$ is p-projection of $(g_n)$ and $(g_n)\in\mathcal{C}$,
then $(f_n)\in\mathcal{C}$.
A sequence $(g_n)$ is called $\mathcal{C}$-{\em complete} iff
$(g_n)\in\mathcal{C}$ and any $(f_n)\in\mathcal{C}$ is a p-projection of~$(g_n)$.

The determinant has the following important
universality property~\cite{vali:79-3,toda:92,mapo:04}:
if $\wsL(f)\le m$ then $f$ is a projection of $\det_{m+1}$.
This implies that the sequence $(\det_n)$ of determinants is $\VPws$-complete~\cite{toda:92}.
Therefore, $\VPe = \VPws$ is equivalent to $(\det_n)\in\VPe$, the major open question mentioned before.
It is not known whether $\VPws$ is different from $\VP$.

We remark that when replacing polynomial upper bounds by quasipolynomial upper bounds
$2^{\log^c n}$ in the definitions of the above three complexity classes,
then all these classes coincide.

We assign now to any of the above complexity classes $\vp_?$ a corresponding
``nondeterministic'' complexity class $\vnp_?$ as follows.
A sequence $(f_n)$ of polynomials belongs to $\vnp_?$ if there
exists a polynomial~$p$ and a sequence $(g_n)\in\vp_?$ such that
$f_n(x)=\sum_{e}g_n(x,e)$ for all~$n$, where the sum is over all $e\in\{0,1\}^{p(n)}$.
It is a nontrivial fact that the resulting classes are the same:
$\vnp_e=\vnp_{\mathrm{ws}}=\vnp$,
for an intuitive proof see~\cite{mapo:04}. Clearly $\VP\subseteq\VNP$.

Valiant~\cite{vali:79-3} proved the major result that $(\per_n)$ is $\VNP$-complete.
Thus $(\per_n)\not\in\VP$ is equivalent to $\VP\ne\VNP$, which is sometimes called
{\em Valiant's hypothesis}.
This can be seen as an algebraic version of Cook's famous
$\mathrm{P}\ne\mathrm{NP}$ hypothesis.
There is great empirical evidence that Valiant's hypothesis is true:
if it were false, then
most of the complexity classes considered by researchers today would collapse~\cite{buer:97-1}.
Proving this implication relies on the generalized Riemann hypothesis,
but we note that the latter can be omitted when dealing with the constant-free versions of
the complexity classes (where only $0,1$ are allowed as constants instead of any
complex numbers).

It is natural to weaken Valiant's hypothesis to $\VPws\ne\VNP$.
In view of the completeness of the sequences of determinants and permanents in $\VPws$ and $\VNP$,
respectively, $\VPws\ne\VNP$ is logically equivalent to the claim that
$(\per_n)$ is not a p-projection of $(\det_n)$.
The latter is a purely mathematical statement, not involving any notions of computation.
This is why some people (including ourselves) believe that this offers one of the most
promising possibilities to attack the P v.s.\  NP problem.

\subsection{Approximate complexity classes}

In \cite{buer:03a} it was proposed to study the notion of approximate complexity in
Valiant's framework. There is a natural way to put a topology on the polynomial ring
$A:=\BC[X_1,X_2,\ldots]$ as a limit of the Euclidean topologies on the finite dimensional subspaces
$\{f\in \BC[X_1,\ldots,X_n]\mid \deg f\le d\}$ whose union over $n,d$ is $A$.

\begin{definition}
The {\em approximate complexity} $\bL(f)$ of~$f\in A$ is defined as the minimum $r\in\N$
such that $f$ is in the closure of $\{g\in A \mid L(g)\le r \}$.
Replacing here $L(g)$ by $\wsL(g)$ we obtain the {\em approximate complexity} $\bwsL(f)$.
\end{definition}

We remark that the same complexity notions are obtained when using the Zariski topology,
since constructible sets have the same closure with respect to Euclidean and Zariski topology.
For more information on approximate complexity we refer to \cite{buer:03a}.

We define the complexity class $\bVPws$ as the set of sequences $(f_n)$ of complex polynomials
such that $\bwsL(f_n)$ is polynomially bounded in~$n$.
Similarly, one defines the classes $\bVP$.
Clearly, $\bVPws\subseteq\bVP$ and both classes are closed under p-projections.
It is not known whether or not $\bVPws$ is contained in $\VNP$.

We go now back to the GCT approach of \cite{MS1}, which attempts to show
Conjecture \ref{msmainconj}.

\begin{proposition}\label{pro:equiv}
Conjecture \ref{msmainconj} is equivalent to $(\per_m) \not\in\bVPws$ and equivalent to $\VNP\not\subseteq \bVPws$.
\end{proposition}

Before giving the proof we note that Conjecture~\ref{msmainconj}
would imply that $\VPws\ne\VNP$ (but not {\it a priori}  $\vp\neq \vnp$).

\begin{proof}
The second equivalence is a consequence of the $\vnp$ completeness of $(\per_m)$.
To show the first equivalence suppose first that Conjecture~\ref{msmainconj}  is false.
Then there exist $c\ge 1$ and $m_0$ such that for all $m\ge m_0$,
$[\ell^{m^c-m}\per_m]$ is contained in the projective orbit closure
$\overline{\GL_{m^{2c}}\cdot [\det_{m^c}]}$ in $\BP(S^{m^c} \BC^{m^{2c}})$.
This implies
$\ell^{m^c-m}\per_m\in \overline{\GL_{m^{2c}}\cdot \det_{m^c}}\subset S^{n^c}\BC^{m^{2c}}$.
Thus for fixed $m\ge m_0$, there exists a sequence $(\sigma_k)$ in $\GL_{m^{2c}}$ such that
$f_k:=\sigma_k\cdot \det_{m^c}$ satisfies
$\lim_{k\to\infty} f_k = \ell^{m^c-m} \per_m$.
There is a  weakly-skew arithmetic circuit for $\det_{m^c}$ of size
polynomial in~$m$.
Composing this circuit with an arithmetic circuit  for
matrix-vector multiplication that
computes the linear transformation~$\sigma_k$
yields a weakly-skew arithmetic circuit for $f_k$ of size at most $m^{c'}$,
where $c'$ denotes a constant (independent of $m,k$).
(In order to preserve the weak-skewness we may need several
copies of the circuit computing the linear transformation $\sigma_k$.)
Let $f'_k$ denote the polynomial obtained from $f_k$
after substituting $\ell$ by $1$ and leaving the variables of $\per_m$ unchanged.
Then $\wsL(f'_k)\le\wsL(f_k)\le m^{c'}$ and $\lim_{k\to\infty} f'_k = \per_m$.
Hence, by definition, we have $\bwsL(\per_m)  \le m^{c'}$ for all $m\ge m_0$,
which implies $(\per_m)\in\bVPws$.

To show the other direction suppose that $(\per_m)\in\bVPws$.
Hence there exists $c\ge 1$ and~$m_0$ such that
$\bwsL(\per_m)< m^c$ for all $m\ge m_0$. Fix $m\ge m_0$ and put $n=m^c$
to ease notation.
By definition, there exists a sequence of forms $f_k$ such that
$\lim_{k\to\infty} f_k = \per_m$ and $\wsL(f_k) < n$ for all~$k$.
The universality of the determinant implies that
$f_k$ is a projection of $\det_{n}$, say
$f_k(x)=\det (M_k)$
where $M_k$ is an $n$ by~$n$ matrix whose entries are affine linear forms
in the variables $x_i$.
We homogenize now with respect to an additional variable~$\ell$:
i.e., we substitute $x_i$ by $x_i/\ell$ and multiply the result by  $\ell^n$.
This implies
$$
 \ell^{n-m} f_k(x) = \ell^{n} f_k (\frac1{\ell}x) = \det (M_k')
$$
with a matrix $M_k'$ whose entries are linear forms in $x_i$ and $\ell$.
Since $\GL_{n^2}$ is dense in $\Mat_{n\times n}$, we conclude that
the form $\ell^{n-m} f_k$ lies in the closure of
$\GL_{n^2}\cdot \det_n$.
As $\lim_{k\to\infty} f_k = \per_m$,
this implies that $\ell^{n-m} \per_m$ lies in the closure of $\det_n$.
This holds for all $m\ge m_0$ with $n=m^c$, so
Conjecture~\ref{msmainconj} would be false.
\end{proof}

\begin{remark}\label{re:ryser}
Using the known fact $L_e(\per_m) = \cO(m^2 2^m)$ from \cite{MR0150048},
the proof of Proposition~\ref{pro:equiv} implies that
$\ol{GL_{n^2} \cdot [\ell^{n-m}\tperm_m]} \subset \ol{GL_{n^2} \cdot [\tdet_n]}$
for $n=\cO(m^2 2^m)$.
\end{remark}

\subsection{Order of approximation}
We   now discuss whether approximation is actually necessary.
Let $R=\BC[[\epsilon]]$ the ring of formal power series in $\epsilon$
and $K$ its quotient field. Substituting $\epsilon$ by $0$
defines the morphism $R\rightarrow\BC, r\mapsto (r)_{\epsilon=0}$
which extends to $S^n R^N\rightarrow S^n \BC^N$.
Note that the group $ \GL_N(K)$
operates on the scalar extension $S^n K^N$ in the natural way.

The following result is due to Hilbert~\cite{hilb:93}.
For a proof we refer to Kraft~\cite[III.2.3, Lemma~1]{MR768181}.

\begin{lemma}
Suppose that $f$ lies in the $GL_N(\BC)$-orbit closure of $g\in S^n\BC^N$.
Then there exists $\sigma\in \GL_N(K)$  such that
$F:=\sigma\cdot g \in S^n R^N$ satisfies $(F)_{\epsilon =0} = f$.
\end{lemma}

Assume we are in the situation of the lemma.
By multiplying with a sufficiently high power of $\epsilon$,
we get  $R$-linear forms $y_1,\ldots,y_N$  such that
\begin{equation}\label{eq:hilbert}
 g(y_1,\ldots,y_N) = \epsilon^q f + \epsilon^{q+1} \tilde{F}
\end{equation}
with some $q\in\N$ and $\tilde{F}\in S^n R^N$.
We then say that $f$ can be approximated with order at most~$q$
along a curve in the orbit of $\det_n$.

\smallskip

\begin{question}\label{OA} Suppose that $f$ lies in orbit closure of $\det_n$ in
$S^n\BC^{n^2}$.
Can the order of approximation of $f$ along a curve in the orbit of $\det_n$
be bounded by a polynomial in~$n$?
\end{question}

In \cite[Thm.~5.7]{buer:03a} an exponential upper bound on the order of approximation
is proven in a more general situation.

We show now that if Question \ref{OA} has an affirmative answer, then approximations can be
eliminated in the context of the GCT-approach.\peterchange

\begin{proposition}\label{pro:OA}
If Question \ref{OA} has an affirmative answer, then $\VPws=\bVPws$.
\end{proposition}

In the present form, this observation is new,
although the proof is similar to the arguments in \cite{buer:03a}.
We make some preparations for the proof.
A {\em skew} arithmetic circuit is an arithmetic circuit such that
for each multiplication gate $\alpha$ at least one of the two vertices pointing to~$\alpha$
is an input vertex. Hence the multiplication is either by a variable or a constant.
It is clear that skew circuits are weakly-skew.
Astonishingly, skew circuits are no less powerful than weakly-skew circuits.
For each weakly-skew circuit there exists a skew circuit with at most double size
that computes the same polynomial, cf. ~\cite{koka:08}.

Let $R=\BC[[\epsilon]]$ and $F\in R[X_1,\ldots,X_N]$.
We denote by $\wsL(F)$ the smallest size of a weakly-skew arithmetic circuit computing~$F$
from the variables $X_i$ and constants in $R$.
Write $F=\sum_i f_i\epsilon^i$ with $f_i\in \BC[X_1,\ldots,X_N]$.

\begin{lemma}\label{le:elima}
We have
$\wsL(f_0,\ldots,f_{q}) =\Oh(q^2 \wsL(F))$
for any $q\in\N$.
\end{lemma}

\begin{proof}
Suppose we have a weakly-skew circuit of size~$s$ computing $F$ from the  variables
and constants $c =\sum_i c_i\epsilon^i\in R$.
By the previous comment we can assume without loss of generality that
the circuit is skew.
Let $g\in R[X_1,\ldots,X_N]$ be an intermediate result of the computation
and write $g=\sum_i g_i\epsilon^i$ with $g_i\in\BC[X_1,\ldots,X_N]$.
The idea is to construct an arithmetic circuit that instead of $g$
computes the coefficients $g_0,\ldots,g_q$ up to degree~$q$
from the variables and the coefficients $c_0,\ldots,c_q$
of the constants~$c$.
This is achieved by replacing each addition of the original circuit
by $q+1$ additions of the corresponding coefficients.
Each multiplication $f=g\cdot h$ of the original circuit is replaced by
$\Oh(q^2)$ arithmetic operations following
$f_k=\sum_{i=0}^kg_i\cdot h_{k-i}$.
This results in a circuit of size $\Oh(sq^2)$.
Since the original circuit is assumed to be skew, it is clear that
the new circuit can be realized by a skew circuit as well.
(We note that it is not obvious how to preserve weak-skewness.)
\end{proof}

\begin{proof}(of Proposition~\ref{pro:OA})
Suppose that $(f_m)\in\bVPws$.
Then $\bwsL(f_m) < n$ with $n$ polynomially bounded in $m$.
Hence $f_m$ is in the closure of the set of polynomials $g$ satisfying
$\wsL(g) < n$.
By the universality of the determinant,
those polynomials~$g$ are projections of $\det_n$,
hence contained in $\ol{\GL_{n^2}\cdot \tdet_n}$.
It follows that $f_m\in \ol{\GL_{n^2}\cdot \tdet_n}$.
If Question \ref{OA} has an affirmative answer, then
$f_m$ can be approximated with order at most~$q$
along a curve in the orbit of $\det_n$, where $q$
is polynomially bounded in~$n$ and hence in $m$.
Hence we are in the situation~(\ref{eq:hilbert})
and have
$$
 F:=\det_n(y_1,\ldots,y_{n^2}) = \epsilon^q f_m + \epsilon^{q+1} \tilde{F}
$$
with $R$-linear forms $y_1,\ldots,y_{n^2}$ in the variables $x_{ij}$
and some polynomial $\tilde{F}$ over $R$ in~$x_{ij}$.
From this we conclude $\wsL(F) = m^{\Oh(1)}$.
Lemma~\ref{le:elima} tells us that
$\wsL(f_m) =\Oh(q^2 \wsL(F))$.
Since $q$ was assumed to be polynomially bounded in~$m$,
we conclude that $\wsL(f_m)$ is polynomially bounded in~$m$ as well.
This implies $(f_m)\in\VPws$.
\end{proof}

\bibliographystyle{amsplain}
\bibliography{Lmatrix}

\def\cdprime{$''$} \def\cprime{$'$} \def\cprime{$'$} \def\cprime{$'$}
  \def\Dbar{\leavevmode\lower.6ex\hbox to 0pt{\hskip-.23ex \accent"16\hss}D}
  \def\cprime{$'$} \def\cprime{$'$} \def\cdprime{$''$}
  \def\Dbar{\leavevmode\lower.6ex\hbox to 0pt{\hskip-.23ex \accent"16\hss}D}
  \def\cprime{$'$} \def\cprime{$'$} \def\cprime{$'$} \def\cprime{$'$}
  \def\Dbar{\leavevmode\lower.6ex\hbox to 0pt{\hskip-.23ex \accent"16\hss}D}
  \def\cprime{$'$} \def\cprime{$'$}
\providecommand{\bysame}{\leavevmode\hbox to3em{\hrulefill}\thinspace}
\providecommand{\MR}{\relax\ifhmode\unskip\space\fi MR }
\providecommand{\MRhref}[2]{%
  \href{http://www.ams.org/mathscinet-getitem?mr=#1}{#2}
}
\providecommand{\href}[2]{#2}
\begin{thebibliography}{10}

\bibitem{MR2264933}
Cristina~M. Ballantine and Rosa~C. Orellana, \emph{A combinatorial
  interpretation for the coefficients in the {K}ronecker product {$s\sb
  {(n-p,p)}\ast s\sb \lambda$}}, S\'em. Lothar. Combin. \textbf{54A} (2005/07),
  Art. B54Af, 29 pp. (electronic). \MR{MR2264933 (2008a:05267)}

\bibitem{MR1750957}
Arkady Berenstein and Reyer Sjamaar, \emph{Coadjoint orbits, moment polytopes,
  and the {H}ilbert-{M}umford criterion}, J. Amer. Math. Soc. \textbf{13}
  (2000), no.~2, 433--466 (electronic). \MR{MR1750957 (2001a:53121)}

\bibitem{berk:84}
S.~Berkowitz, \emph{On computing the determinant in small parallel time using a
  small number of processors}, Information Processing Letters \textbf{18}
  (1984), 147--150.

\bibitem{botta:67}
Peter Botta, \emph{Linear transformations that preserve the permanent}, Proc.
  Amer. Math. Soc. \textbf{18} (1967), 566--569. \MR{MR0213376 (35 \#4240)}

\bibitem{bren:74}
R.P. Brent, \emph{The complexity of multiprecision arithmetic}, Proc.\ Seminar
  on Compl.\ of Comp.\ Problem Solving, Brisbane, 1975, pp.~126--165.

\bibitem{BORpreprint}
E~Briand, R.~Orellana, and M.~Rosas, \emph{Reduced {K}ronecker coefficients and
  counter-examples to {M}ulmuley's saturation conjecture {S}{H}}, preprint
  arXiv:0810.3163v1 (2008).

\bibitem{MR1243152}
Michel Brion, \emph{Stable properties of plethysm: on two conjectures of
  {F}oulkes}, Manuscripta Math. \textbf{80} (1993), no.~4, 347--371.
  \MR{MR1243152 (95c:20056)}

\bibitem{buer:00-3}
P.~B{\"u}rgisser, \emph{Completeness and reduction in algebraic complexity
  theory}, {A}lgorithms and {C}omputation in {M}athematics, vol.~7, Springer
  Verlag, 2000.

\bibitem{buer:97-1}
\bysame, \emph{{C}ook's versus {V}aliant's hypothesis}, Theoretical Computer
  Science \textbf{235} (2000), 71--88.

\bibitem{buer:03a}
\bysame, \emph{The complexity of factors of multivariate polynomials},
  Foundations of {C}omputational {M}athematics \textbf{4} (2004), 369--396.

\bibitem{buci:10}
P.~B\"urgisser, M.~Christandl, and C.~Ikenmeyer, \emph{Even partitions in
  plethysms}, Accepted for {\em J. Algebra}. arXiv 1003.4474v1 (2010).

\bibitem{buci:09}
\bysame, \emph{Nonvanishing of {K}ronecker coefficients for rectangular
  shapes}, arXiv 0910.4512v2 (2009).

\bibitem{buic:10}
P.~B\"urgisser and C.~Ikenmeyer, \emph{Geometric complexity theory and tensor
  rank}, arXiv:1011.1350 (2010).

\bibitem{MR2276458}
Matthias Christandl, Aram~W. Harrow, and Graeme Mitchison, \emph{Nonzero
  {K}ronecker coefficients and what they tell us about spectra}, Comm. Math.
  Phys. \textbf{270} (2007), no.~3, 575--585. \MR{MR2276458 (2007k:20029)}

\bibitem{eisenbud}
David Eisenbud, \emph{Commutative algebra}, Graduate Texts in Mathematics, vol.
  150, Springer-Verlag, New York, 1995, With a view toward algebraic geometry.
  \MR{MR1322960 (97a:13001)}

\bibitem{MR954659}
David Eisenbud and Joe Harris, \emph{Vector spaces of matrices of low rank},
  Adv. in Math. \textbf{70} (1988), no.~2, 135--155. \MR{MR954659 (89j:14010)}

\bibitem{MR1923785}
Matthias Franz, \emph{Moment polytopes of projective {$G$}-varieties and tensor
  products of symmetric group representations}, J. Lie Theory \textbf{12}
  (2002), no.~2, 539--549. \MR{MR1923785 (2003j:20077)}

\bibitem{Frobdet}
G.~Frobenius, \emph{{\"U}ber die {D}arstellung der endlichen {G}ruppen durch
  lineare {S}ubstitutionen}, Sitzungsber Deutsch. Akad. Wiss. Berlin (1897),
  994--1015.

\bibitem{FH}
William Fulton and Joe Harris, \emph{Representation theory}, Graduate Texts in
  Mathematics, vol. 129, Springer-Verlag, New York, 1991, A first course,
  Readings in Mathematics. \MR{MR1153249 (93a:20069)}

\bibitem{MR0414794}
David~A. Gay, \emph{Characters of the {W}eyl group of {$SU(n)$} on zero weight
  spaces and centralizers of permutation representations}, Rocky Mountain J.
  Math. \textbf{6} (1976), no.~3, 449--455. \MR{MR0414794 (54 \#2886)}

\bibitem{MR2072621}
A.~{\`E}. Guterman and A.~V. Mikhal{\"e}v, \emph{General algebra and linear
  mappings that preserve matrix invariants}, Fundam. Prikl. Mat. \textbf{9}
  (2003), no.~1, 83--101. \MR{MR2072621 (2005f:15002)}

\bibitem{MR0282977}
Robin Hartshorne, \emph{Ample subvarieties of algebraic varieties}, Notes
  written in collaboration with C. Musili. Lecture Notes in Mathematics, Vol.
  156, Springer-Verlag, Berlin, 1970. \MR{MR0282977 (44 \#211)}

\bibitem{hilb:93}
D.~Hilbert, \emph{\"{U}ber die vollen {I}nvariantensysteme}, Math.~Ann.
  \textbf{42} (1893), 313--373.

\bibitem{MR0199184}
Heisuke Hironaka, \emph{Resolution of singularities of an algebraic variety
  over a field of characteristic zero. {I}, {II}}, Ann. of Math. (2) 79 (1964),
  109--203; ibid. (2) \textbf{79} (1964), 205--326. \MR{MR0199184 (33 \#7333)}

\bibitem{MR1321638}
Roger Howe, \emph{Perspectives on invariant theory: {S}chur duality,
  multiplicity-free actions and beyond}, The {S}chur lectures (1992) ({T}el
  {A}viv), Israel Math. Conf. Proc., vol.~8, Bar-Ilan Univ., Ramat Gan, 1995,
  pp.~1--182. \MR{MR1321638 (96e:13006)}

\bibitem{koka:08}
E.~Kaltofen and P.~Koiran, \emph{Expressing a fraction of two determinants as a
  determinant}, Proc.\ {ISSAC '08}, ACM, 2008.

\bibitem{MR506989}
George~R. Kempf, \emph{Instability in invariant theory}, Ann. of Math. (2)
  \textbf{108} (1978), no.~2, 299--316. \MR{MR506989 (80c:20057)}

\bibitem{Klyachkopreprint}
A.~Klyachko, \emph{Quantum marginal problem and representations of the
  symmetric group}, preprint arXiv:quant-ph/0409113v1 (2004).

\bibitem{MR1920389}
Anthony~W. Knapp, \emph{Lie groups beyond an introduction}, second ed.,
  Progress in Mathematics, vol. 140, Birkh\"auser Boston Inc., Boston, MA,
  2002. \MR{MR1920389 (2003c:22001)}

\bibitem{MR0158024}
Bertram Kostant, \emph{Lie group representations on polynomial rings}, Amer. J.
  Math. \textbf{85} (1963), 327--404. \MR{MR0158024 (28 \#1252)}

\bibitem{MR768181}
Hanspeter Kraft, \emph{Geometrische {M}ethoden in der {I}nvariantentheorie},
  Aspects of Mathematics, D1, Friedr. Vieweg \& Sohn, Braunschweig, 1984.
  \MR{MR768181 (86j:14006)}

\bibitem{kumar:10}
Shrawan Kumar, \emph{Geometry of orbits of permanents and determinants},
  arXiv:1007.1695v1 (2010).

\bibitem{MR1923198}
Shrawan Kumar, \emph{Kac-{M}oody groups, their flag varieties and
  representation theory}, Progress in Mathematics, vol. 204, Birkh\"auser
  Boston Inc., Boston, MA, 2002. \MR{MR1923198 (2003k:22022)}

\bibitem{LMRdet}
J.M. Landsberg, Laurent Manivel, and Ressayre Nickolas, \emph{Dual varieties
  and the gct program}, preprint (2010).

\bibitem{macdonald}
I.~G. Macdonald, \emph{Symmetric functions and {H}all polynomials}, second ed.,
  Oxford Mathematical Monographs, The Clarendon Press Oxford University Press,
  New York, 1995, With contributions by A. Zelevinsky, Oxford Science
  Publications. \MR{MR1354144 (96h:05207)}

\bibitem{MR1603309}
Kay Magaard and Gunter Malle, \emph{Irreducibility of alternating and symmetric
  squares}, Manuscripta Math. \textbf{95} (1998), no.~2, 169--180. \MR{1603309
  (99a:20011)}

\bibitem{mapo:04}
G.~Malod and N.~Portier, \emph{Characterizing {V}aliant's algebraic complexity
  classes}, Journal of Complexity \textbf{24} (2008), 16--38.

\bibitem{arXiv:0809.3710}
Laurent Manivel, \emph{A note on certain {K}ronecker coefficients}, preprint
  arXiv:0809.3710v1.

\bibitem{MR1465785}
\bysame, \emph{Applications de {G}auss et pl\'ethysme}, Ann. Inst. Fourier
  (Grenoble) \textbf{47} (1997), no.~3, 715--773. \MR{MR1465785 (98h:20078)}

\bibitem{MR1651092}
\bysame, \emph{Gaussian maps and plethysm}, Algebraic geometry ({C}atania,
  1993/{B}arcelona, 1994), Lecture Notes in Pure and Appl. Math., vol. 200,
  Dekker, New York, 1998, pp.~91--117. \MR{MR1651092 (99h:20070)}

\bibitem{MR0137729}
Marvin Marcus and F.~C. May, \emph{The permanent function}, Canad. J. Math.
  \textbf{14} (1962), 177--189. \MR{MR0137729 (25 \#1178)}

\bibitem{MR0109854}
Yoz{\^o} Matsushima, \emph{Espaces homog\`enes de {S}tein des groupes de {L}ie
  complexes}, Nagoya Math. J \textbf{16} (1960), 205--218. \MR{MR0109854 (22
  \#739)}

\bibitem{MR504978}
Henryk Minc, \emph{Permanents}, Encyclopedia of Mathematics and its
  Applications, vol. 9999, Addison-Wesley Publishing Co., Reading, Mass., 1978,
  With a foreword by Marvin Marcus, Encyclopedia of Mathematics and its
  Applications, Vol. 6. \MR{MR504978 (80d:15009)}

\bibitem{MS8}
Ketan~D. Mulmuley, \emph{Geometric complexity theory: On canonical bases for
  the nonstandard quantum groups}, preprint.

\bibitem{MS6}
\bysame, \emph{Geometric complexity theory {VI}: the flip via saturated and
  positive integer programming in representation theory and algebraic
  geometry,}, Technical Report TR-2007-04, computer science department, The
  University of Chicago, May, 2007.

\bibitem{MS7}
\bysame, \emph{Geometric complexity theory {VII}: Nonstandard quantum group for
  the plethysm problem}, preprint.

\bibitem{MS5}
Ketan~D. Mulmuley and H.~Narayaran, \emph{Geometric complexity theory {V}: On
  deciding nonvanishing of a generalized {L}ittlewood-{R}ichardson
  coefficient}, Technical Report TR-2007-05, computer science department, The
  University of Chicago, May, 2007.

\bibitem{MS3}
Ketan~D. Mulmuley and Milind Sohoni, \emph{Geometric complexity theory {III}:
  on deciding positivity of {L}ittlewood-{R}ichardson coefficients}, preprint
  cs.ArXiv preprint cs.CC/0501076.

\bibitem{MS4}
\bysame, \emph{Geometric complexity theory {IV}: quantum group for the
  {K}ronecker problem}, preprint available at UC cs dept. homepage.

\bibitem{MS1}
\bysame, \emph{Geometric complexity theory. {I}. {A}n approach to the {P} vs.\
  {NP} and related problems}, SIAM J. Comput. \textbf{31} (2001), no.~2,
  496--526 (electronic). \MR{MR1861288 (2003a:68047)}

\bibitem{MS2}
\bysame, \emph{Geometric complexity theory. {II}. {T}owards explicit
  obstructions for embeddings among class varieties}, SIAM J. Comput.
  \textbf{38} (2008), no.~3, 1175--1206. \MR{MR2421083}

\bibitem{mur:38}
F.~D. Murnaghan, \emph{The {A}nalysis of the {K}ronecker {P}roduct of
  {I}rreducible {R}epresentations of the {S}ymmetric {G}roup}, Amer. J. Math.
  \textbf{60} (1938), no.~3, 761--784. \MR{MR1507347}

\bibitem{MR0396847}
A.~M. Popov, \emph{Irreducible simple linear {L}ie groups with finite standard
  subgroups in general position}, Funkcional. Anal. i Prilo\v zen. \textbf{9}
  (1975), no.~4, 81--82. \MR{MR0396847 (53 \#707)}

\bibitem{Popovpreprint}
V.~Popov, \emph{Two orbits: {W}hen is one in the closure of the other?},
  preprint arXiv:0808.2735v7 (2008).

\bibitem{MR2265844}
Claudio Procesi, \emph{Lie groups}, Universitext, Springer, New York, 2007, An
  approach through invariants and representations. \MR{MR2265844 (2007j:22016)}

\bibitem{resspreprint}
N.~Ressayre, \emph{Theory and generalized eigenvalue problem}, preprint
  arXiv:0704.2127.

\bibitem{MR0150048}
Herbert~John Ryser, \emph{Combinatorial mathematics}, The Carus Mathematical
  Monographs, No. 14, Published by The Mathematical Association of America,
  1963. \MR{MR0150048 (27 \#51)}

\bibitem{MR0026286}
C.~E. Shannon, \emph{A mathematical theory of communication}, Bell System Tech.
  J. \textbf{27} (1948), 379--423, 623--656. \MR{MR0026286 (10,133e)}

\bibitem{MR1800533}
A.~Skowro{\'n}ski and J.~Weyman, \emph{The algebras of semi-invariants of
  quivers}, Transform. Groups \textbf{5} (2000), no.~4, 361--402. \MR{MR1800533
  (2001m:16017)}

\bibitem{MR0447428}
T.~A. Springer, \emph{Invariant theory}, Lecture Notes in Mathematics, Vol.
  585, Springer-Verlag, Berlin, 1977. \MR{MR0447428 (56 \#5740)}

\bibitem{MR2049555}
Richard~P. Stanley, \emph{Irreducible symmetric group characters of rectangular
  shape}, S\'em. Lothar. Combin. \textbf{50} (2003/04), Art. B50d, 11 pp.
  (electronic). \MR{MR2049555 (2005e:20020)}

\bibitem{toda:92}
S.~Toda, \emph{Classes of arithmetic circuits capturing the complexity of
  computing the determinant}, IEICE Trans. Inf. Syst. \textbf{E75-D} (1992),
  116--124.

\bibitem{vali:79-3}
Leslie~G. Valiant, \emph{Completeness classes in algebra}, Proc.~11th ACM STOC,
  1979, pp.~249--261.

\bibitem{vali:82}
L.G. Valiant, \emph{Reducibility by algebraic projections}, Logic and
  Algorithmic: an International Symposium held in honor of Ernst Speck\-er,
  vol.~30, Monogr.\ {N}o.~30 de l'{E}nseign.\ Math., 1982, pp.~365--380.

\bibitem{MR922386}
Joachim von~zur Gathen, \emph{Feasible arithmetic computations: {V}aliant's
  hypothesis}, J. Symbolic Comput. \textbf{4} (1987), no.~2, 137--172.
  \MR{MR922386 (89f:68021)}

\bibitem{weyman}
Jerzy Weyman, \emph{Cohomology of vector bundles and syzygies}, Cambridge
  Tracts in Mathematics, vol. 149, Cambridge University Press, Cambridge, 2003.
  \MR{MR1988690 (2004d:13020)}

\end{thebibliography}
\end{document}